\documentclass[11pt]{article}
\usepackage[letterpaper]{geometry}
\usepackage{amsmath,amsthm,amsfonts,amssymb}
\usepackage{enumerate,color}
\usepackage{graphicx}
\usepackage{subfigure}
\usepackage{comment}
\usepackage{url}
\usepackage{natbib}

\numberwithin{equation}{section}
\theoremstyle{plain}

\newtheorem{theorem}{Theorem}

\newtheorem{corollary}[theorem]{Corollary}

\newtheorem{proposition}[theorem]{Proposition}

\newtheorem{remark}[theorem]{Remark}

\begin{document}

\begin{center}
  \Large \bf Transform Analysis for Hawkes Processes with Applications in Dark Pool Trading
\end{center}

\author{}
\begin{center}
{Xuefeng
  Gao}\,\footnote{Corresponding author. Department of Systems
    Engineering and Engineering Management, The Chinese University of Hong Kong, Shatin, N.T. Hong Kong;
    xfgao@se.cuhk.edu.hk},
    {Xiang Zhou}\,\footnote{Department of Systems
    Engineering and Engineering Management, The Chinese University of Hong Kong, Shatin, N.T. Hong Kong;
    zhouxng@se.cuhk.edu.hk},
  Lingjiong Zhu\,\footnote{Department of Mathematics, Florida State University, 1017 Academic Way, Tallahassee, FL-32306, United States of America; zhu@math.fsu.edu.
  }
\end{center}

\begin{abstract}
Hawkes processes are a class of simple point processes that are self-exciting and have
clustering effect, with wide applications in finance, social networks
and many other fields. This paper considers a self-exciting Hawkes process where
the baseline intensity is time-dependent, the exciting function is a general function and the jump sizes of the intensity process are independent and identically distributed non-negative random variables. This Hawkes model is non-Markovian in general. We obtain closed-form formulas for the Laplace transform, moments and the distribution of the Hawkes process. To illustrate the applications of our results, we use the Hawkes process to model the clustered arrival of trades in a dark pool and analyze various performance metrics including time-to-first-fill, time-to-complete-fill and the expected fill rate of a resting dark order.
\end{abstract}

%

\section{Introduction}

Consider a positive sequence of event arrival times $\tau_{1}<\tau_{2}<\cdots$,
that are defined on a complete probability space $(\Omega,\mathcal{F},\mathbb{P})$
with right-continuous and complete
information filtration $(\mathcal{F}_{t})_{t\geq 0}$.
We define a counting process $N$ and an associated point process $L$ as
\begin{equation*}
N_{t}=\sum_{n=1}^{\infty}1_{\tau_{n}\leq t} \quad \text{and} \quad L_{t}=\sum_{n=1}^{\infty}\ell_{n}\cdot 1_{\tau_{n}\leq t},
\end{equation*}
where $\{ \ell_{n}: n \ge 1 \}$ is a sequence of independent and identically distributed (i.i.d.) non-negative random variables, and $\ell_{n}$ is $\mathcal{F}_{\tau_{n}}$-measurable for each $n\in\mathbb{N}$.

We consider $\{N_{t}: t \ge 0\}$ to be a Hawkes process with random jump sizes in the intensity, that is a simple
point process $N$ with a stochastic intensity given by
\begin{equation}
\lambda_{t}=\mu(t)+\int_{0}^{t-}h(t-s)dL_{s} = \mu(t) + \sum_{0<\tau_i<t}  h(t-\tau_i) \cdot \ell_i , \label{eq:dynamics}
\end{equation}
where $\mu(\cdot)\geq 0$ is the time-dependent baseline intensity and $h(\cdot):\mathbb{R}_{\geq 0}\rightarrow\mathbb{R}_{\geq 0}$
is the exciting function encoding the influence of past events on the intensity
and we always assume that $\Vert h\Vert_{L^{1}}=\int_{0}^{\infty}h(t)dt<\infty$
and $h$ is locally bounded.
In the theory of point processes \citep{Daley}, the random jump sizes $\ell_{i}$ are sometimes referred to
as the random marks associated with the point process, and the point process
with intensity \eqref{eq:dynamics} is a marked Hawkes process.

Two special cases of this Hawkes model have been well studied in the literature. First,
when $\ell_i \equiv 1$ for each $i$, the counting process $N$ is the classical linear Hawkes process introduced by A.G. Hawkes in 1971 \citep{Hawkes, Hawkes71II}.
Hawkes process exhibits both self--exciting (i.e., the occurrence
of an event increases the probabilities of future
events) and clustering properties. It generalizes the standard Poisson process. Hence Hawkes process is very appealing in point process modeling and it has wide applications
in finance. This includes modeling of clustering behavior in stock trade arrivals, default clustering in portfolio credit risk and financial contagion, high-frequency stock prices, etc.
See, e.g., \citet{ZhuThesis, Bacry2015, Jaisson} and references therein for details.

Second, when the exciting function $h$ is exponential, i.e.,
$h(t)= \delta e^{-\kappa t}$ for $t \ge 0$, where $\delta, \kappa >0$, \citet{Errais} studied the transforms and distributions of this Hawkes process with i.i.d. jumps $\{\ell_{i}\}$ and a special time-dependent baseline intensity in the form of $\mu(t) = \mu+e^{-\kappa t} (\lambda_0 -\mu)$. In this case,
the two-dimensional process $(\lambda, N)$
is Markovian. \citet{Errais} used this Markovian Hawkes process to model the clustering of corporate defaults, where the random jump times $\tau_i$ represent default times, and the intensity jump magnitudes $\ell_i$ represent the random losses at default. In particular, the intensity model \eqref{eq:dynamics} captures the empirical feature that the larger the financial loss of a defaulted firm, the larger the impact of such a event on the other firms, and the
bigger the increase of the default intensity at an event.
Relying on the Dynkin formula, the authors of \citet{Errais} characterized the Fourier transform and the distribution of the Hawkes process using ODEs, and they apply these results in a range of applications in portfolio credit risk, including the valuation, hedging and calibration of portfolio credit derivatives.

This paper considers a Hawkes process with intensity in \eqref{eq:dynamics} where the exciting function $h$ is a general function,
the baseline intensity is time-dependent, and the random jump sizes $\{ \ell_{n}: n \ge 1 \}$ are i.i.d. nonnegative random variables.
We pursue this extended Hawkes model for two reasons: first, we would like to extend the transform analysis of Markovian Hawkes processes in \citet{Errais} to
the general setting which allows a general time--dependent baseline intensity to account for non--stationarity such as intraday seasonalities in trading activities and non-exponential exciting functions to account for possibly non-Markovian dynamics;
second, our motivating application in dark pool trading, which will be illustrated later, naturally fits this general Hawkes model.

In our setting, the Hawkes process can be non-Markovian as a result of the general exciting function $h(\cdot)$.
Relying on the immigration-birth representation of linear Hawkes processes given in \citet{HawkesII}, and
in particular \citet{Karabash} for marked linear Hawkes processes, we obtain closed-form formulas for the Laplace transform, moments and the distribution of the Hawkes process $(N, L)$ via integral equations. In the special case of an exponential exciting function, we recover the results obtained in \citet{Errais}.

The closed-form formulas of transforms and the probability distribution of Hawkes processes generate computational tractability, and they provide insights into the behavior of Hawkes processes.
They could be useful in applications in finance and other fields where event occurrences exhibit self-exciting and clustering. In this paper,
we apply our theoretical results to analyze the performance of dark pools.

Dark pools are automated trading facilities which do not display bid and ask quotes to the public, hence they can be used to reduce the
market impact of trading big orders. There are around 40 active dark pools in the U.S. for equity trading. Dark pools now account for about 15\% of the trading volume in the U.S. equity market and about $7\%$ in Europe.
See, e.g., \citet{Mittal2008, ZhuHX2014} for an overview.
We focus on a typical ``midpoint"  dark pool using a continuous matching mechanism, where participants
submit buy or sell orders with specified quantities for a particular security.
Trades can occur at any time if there is
liquidity on both sides of the market, and the matching price is the midpoint of the best bid and offer on transparent exchanges. If an
investor rests an order in a pool for some time and the order is not completely filled, then the remaining quantity may be cancelled and
submitted to a different dark pool or an exchange to seek liquidity.

Several theoretical and empirical studies have suggested that the liquidity in dark pool is clustered, i.e., ``liquidity begets liquidity". This
means that once a trade has occurred in a dark pool, the probability of observing another one increases. See, e.g.,
\citet{Buti2011}, Chapter~3 in \citet{Lehalle2013} and \citet{Markov2013} for details. Various market events can lead to trade clustering in dark pools.
For example, an institutional investor who trades and gets a fill from a particular pool can re-route his orders from another venue back to this pool. In addition, high
frequency traders in the market who are fishing in the dark pool may also notice the existence of a big order from a trade occurrence and
they may also come to trade in this pool \citep{Mittal2008}. The clustering of liquidity suggests that strategic traders form liquidity expectations from either their own trades or post-trade information even in the absence of pre-trade market transparency, and this allows them to design liquidity seeking algorithms that exploit the clustered arrivals of liquidity to maximize the fill rate of their orders. It also suggests that in fragmented markets, orders can migrate quickly from one venue to another.
A natural model to capture the clustering behavior of trade arrivals in dark pools is the Hawkes process.
Indeed, the classical Hawkes process with $\ell_n \equiv 1$ has been widely used to model
clustering of trade arrivals on transparent exchanges in the literature.
See, e.g., \citet{Bowsher2007, BacryMuzy2014, Cartea2014, AJ, Bacry2015} and references therein.

We consider an investor who rests a large midpoint peg (buy) order in a given dark venue, where the execution price of the order floats with the market at the mid-quote derived from transparent exchanges. As in previous studies \citep{Afeche2014, Kratz2015}, we consider a time-priority rule where orders from counterparties
are matched on a first-come-first-served basis\footnote{Matching rules or allocation mechanisms of dark pools are typically complex, partly confidential and frequently updated \citep{Ye2011}. Time-priority matching rule is used by, e.g., BATS Europe Dark Book, see \citet{Liquidmetrix}. Besides time-priority matching, many dark pools use some form of pro-rata matching \citep{ZhuHX2014}. This matching rule is different from the model we consider here and we leave the study of it for future research.}. We model the execution process
of the investor's resting midpoint order by a Hawkes process $(N, L)$ where $\{\tau_i\}$ represent the arrival times of the consolidated trades (eligible-to-match sell orders) from other players in
the pool and the random variables $\{\ell_i\}$ represent the sizes of arriving trades which may not be a constant. Empirically, it has been
observed that the distribution of resting liquidity in dark pools has fatter tails than exponential distributions. See, e.g., \citet{Ganchev2010}. This
implies that the larger the size of a trade, the more likely it is that there is more quantity remaining in the pool. Hence, liquidity seekers or
high frequency traders may be attracted to put more dark orders to the pool after a trade's occurrence, leading to a bigger increase of the trading intensity at a
trade's occurrence.  Such a feature of positive liquidity feedback could be captured by the self-exciting intensity model \eqref{eq:dynamics}.
In the special case when $h \equiv 0,$ the self-exciting behavior disappears and the point process $L$ modeling the cumulative arriving volume of dark trades
reduces to a compound Poisson process. For tractability purposes, in this paper we do not consider other order attributes such as limit price or minimum execution size which can be attached to a midpoint order as anti-gaming and risk management tools.

Using the transform formulas we obtain for the Hawkes model $(N, L)$, we can efficiently compute performance quantities including time-to-first-fill, time-to-complete-fill, and the expected fill rate in a given time window for a midpoint peg order placed at an empty dark pool. We also analyze the probability of obtaining another fill and the
expected fill size conditioned on there is an initial fill of the midpoint order, to understand liquidity expectations after an occurrence of a trade. Furthermore, we extend our analysis to study non-empty dark
pools. The performance quantities we
compute represent major performance characteristics of dark pools around liquidity \citep{Mittal2007, Afeche2014}.
They could help give investors a guide to maximize fills and liquidity
opportunities from dark pools, and indicate whether and where to trade in a fragmented financial market with multiple dark pools.
Hence, such performance quantities are important for smart order routing and allocation of liquidity among different pools to reduce
market impact and execution costs in portfolio trading. See, e.g., \citet{Mittal2007, Ganchev2010, Laruelle2011} for
detailed discussions.

\bigskip
\textbf{Related literature.} Two streams of research that are closely related to our work are Hawkes processes and dark pools. We now explain the difference between our study and the existing literature in these two areas.

\textit{Hawkes processes.}
The majority of the works on Hawkes processes in the literature
assume a constant baseline intensity $\mu(\cdot)\equiv\mu$.
The case when the baseline intensity and/or the exciting function
are time-dependent is much less studied. In a recent work, 
\citet{Euch2016} obtained the characteristic function of a multivariate Hawkes process $N$ with a time-dependent baseline intensity. They did not consider random jump sizes in the intensity. \citet{RSS} studied the properties of a locally stationary Hawkes process
with both the baseline intensity and exciting function being time-dependent.
See also \citet{Toke} for the estimations of Hawkes processes
with time-dependent baseline intensities and \citet{KL} with time-dependent exciting function
and zero baseline intensity for various applications. Both \citet{Toke} and \citet{RSS} also
used constant jump sizes, while \citet{KL} considered random jump sizes.

Several papers have considered the Hawkes process where the intensity process has random jump sizes as our paper.
Almost all of them remain in the Markovian framework.
In \citet{Dassios2011}, the authors studied a dynamic contagion process by combining the Markovian Hawkes model with i.i.d. intensity jump sizes with externally-excited jumps. They characterized distributional properties of this new process. 
\citet{Errais} and \citet{Zhang2015} studied generalized Markovian Hawkes processes, or affine point processes, where the intensity is an affine function of an affine jump-diffusion.
These models belong to the class of affine processes studied in \citet{Duffie}.
In all these works, the (generalized) Hawkes models are still Markovian. One work that deviates from the Markovian framework,
with time-dependent baseline intensity and random jump sizes, similar as this paper, is \citet{Lee2016}, where the jump size of the intensity is modulated by a stochastic process described by a stochastic differential equation. They proposed
new simulation and model fitting algorithms for the Hawkes model, but they did not obtain distributional properties.
The special case $\mu(t)\equiv\mu$ of our model also belongs to the class of the Hawkes process
with random marks, see. e.g. \citet{Bremaud2002} who studied
the power spectrum, and \citet{Karabash} who studied the limit theorems
and we refer to Section 2.1.1 of \citet{Bacry2015} for more references.

\textit{Dark pools.} In the dark pool literature, our work is closely related to studies including \citet{Markov2013} and \citet{Afeche2014}. The paper \citet{Markov2013} from the industry explicitly modeled the clustering of trade arrivals in a dark pool using the classical Hawkes process with $\ell_n \equiv 1$. They discussed estimation of this classical Hawkes model using exponential exciting functions. \citet{Afeche2014} used a double--sided queueing model to study the operational characteristics of dark pools. They considered Poisson order arrivals and obtained closed-form results for system-level and order-level performances such as fill rates and system times.
Our work focuses on the order-level performance, i.e., the experience of a single resting midpoint order placed at a dark venue. We consider more general Hawkes arrival process to capture the clustering behavior of order arrivals. Incorporating Hawkes processes to study system-level performance of dark pools is left for future work. Our work also complements other studies on dark pools, see, e.g. \citet{Ganchev2010, Laruelle2011, Almgren} for order routing algorithms among multiple pools, \citet{Klock2011, gatheral2013, Kratz2014, Kratz2015} for optimal portfolio trading strategies and price manipulation issues in the presence of a dark pool and a lit exchange,
\citet{Hendershott2000} for the conditions under which investors should use a dark pool versus a traditional trading venue, and \citet{Buti2011, ZhuHX2014, Iyer2015} for effects of dark pool trading on the market quality and welfare analysis.

\textbf{Organization of this paper.}
The rest of the paper is organized as follows. In Section~\ref{sec:2}, we state the main result on the joint Laplace transform of the Hawkes model $(N_T, L_T)$ for a fixed $T>0$. Relying on this result, we obtain explicit formulas for the first two moments of $N_{T}$ and $L_{T}$. We also compute analytically
 the probability mass function
of $N_{T}$ and also that of $L_{T}$ when the jump sizes $\{\ell_i\}$
are lattice distributed.
In Section~\ref{sec:application}, we apply the main results to analyze performance problems arising from trading in dark pools. Section~\ref{sec:conclusion} concludes.
Some technical proofs are collected in the Appendix.

\section{Main results} \label{sec:2}

In this section we present the main results. Throughout this section, we use $\mathbb{C}$ to denote the set of complex numbers, $\mathcal{R}(\theta)$ to denote the real part of a complex number $\theta \in \mathbb{C}$, and $|\theta|$ to denote its modulus.

The key mathematical result is the following joint Laplace transform of the Hawkes process $(N_T, L_T)$ for fixed $T>0.$

\begin{theorem}\label{MainThm}
For any $\theta_{1},\theta_{2}\in\mathbb{C}$ with $\mathcal{R}(\theta_{1})\geq 0$, $\mathcal{R}(\theta_{2})\geq 0$,
\begin{equation}\label{eq:laplace}
\mathbb{E}[e^{-\theta_{1}N_{T}-\theta_{2}L_{T}}]
=e^{\int_{0}^{T}\mu(T-s)(F(s)-1)ds},
\end{equation}
where the function $F$ is the unique solution to the integral equation
\begin{equation} \label{eq:F}
F(t)=e^{-\theta_{1}}\mathbb{E}\left[e^{-\theta_{2}\ell_{1}+\int_{0}^{t}\ell_{1}h(s)(F(t-s)-1)ds}\right],
\end{equation}
with $|F(t)| \le 1$ for $t \in [0, T]$.
\end{theorem}

The Equation~\eqref{eq:F} is a Hammerstein--type Volterra integral equation, and it can be quickly solved numerically using, for example, piecewise polynomial collocation methods. See e.g. Chapter~2 in \citet{Brunner2004} for further details.

\begin{remark}
We show in this remark that we recover the transform of Hawkes processes in \citet{Errais} for an exponential exciting function.
Note when $h(x)= \delta e^{-\kappa x}$ with $\delta, \kappa>0$, \citet{Errais} derived that (Proposition 2.2 in their paper)
with a baseline intensity $\mu(t) = \mu + e^{- \kappa t} (\lambda_{0}-\mu)$ where $\lambda_0 \ge \mu >0$,
\begin{equation*}
\mathbb{E}[e^{-\theta_{1}N_{T}-\theta_{2}L_{T}}]
=\exp \left( B(T) + \lambda_0 \cdot A(T) \right),
\end{equation*}
where the functions $A(\cdot)$ and $B(\cdot)$ satisfy the ODEs
\begin{align}
A'(t) &= - \kappa A(t) -1 + f(\delta A(t) - \theta_2) e^{-\theta_1},\label{AtEqn}\\
B'(t) &= \kappa \mu A(t), \label{eq:Bt}
\end{align}
with $A(0)=B(0)=0$, and $f$ is defined by
$f(\omega) := \mathbb{E} [e^{\omega \cdot \ell_1}]$ for $\omega \in \mathbb{C}$.
Thus using \eqref{eq:Bt} we obtain
\begin{equation}\label{eq:Gie}
\mathbb{E}[e^{-\theta_{1}N_{T}-\theta_{2}L_{T}}]
=\exp \left(\mu \int_0^T \kappa A(t)dt + \lambda_0 \cdot A(T) \right).
\end{equation}
We prove that our result is consistent with the result \eqref{eq:Gie} from \citet{Errais}.
To see this, we first note that when $h(x)= \delta e^{-\kappa x}$, we obtain from Theorem~\ref{MainThm} that the function $F$ is given by
\begin{align}
F(t)&=e^{-\theta_{1}}\mathbb{E}\left[e^{-\theta_{2}\ell_{1}+\int_{0}^{t}\ell_{1} \delta e^{-\kappa (t-s)} (F(s)-1)ds}\right] \nonumber\\
&=e^{-\theta_{1}} f\left( -\theta_{2} + \int_{0}^{t}  \delta e^{-\kappa (t-s)} (F(s)-1)ds\right). \label{eq:F1}
\end{align}
In view of \eqref{eq:laplace}, \eqref{eq:Gie} and the expression of the baseline intensity $\mu(t)$, it suffices to show that
\begin{equation}\label{eq-siam}
\mu \int_0^T \kappa A(t)dt + \lambda_0 \cdot A(T) =  \int_{0}^{T}(\mu + e^{- \kappa (T-s)} (\lambda_{0}-\mu))(F(s)-1)ds.
\end{equation}
To this end, we first prove that for $t \in [0, T],$
\begin{equation}
F(t) -1
= \kappa A(t) + A'(t)\label{FA}
= -1 + f(\delta A(t) - \theta_2) e^{-\theta_1},
\end{equation}
where the second equality above is due to \eqref{AtEqn}.
That is, we need to show for $t \in [0, T],$
\begin{equation} \label{eq:F2}
F(t) = f(\delta A(t) - \theta_2) e^{-\theta_1}.
\end{equation}
In view of \eqref{eq:F} and the fact that $A(0)=0$, this equation clearly holds when $t=0$.
Let us verify that \eqref{eq:F2} is indeed the unique solution for \eqref{eq:F1}. We write for $t \in [0, T],$
\begin{equation} \label{eq:xt}
x(t) := \int_{0}^{t}  e^{-\kappa (t-s)} (F(s)-1)ds.
\end{equation}
Taking derivative at both sides, we find that
\begin{equation*}
x'(t)= -\kappa x(t) + F(t) -1.
\end{equation*}
Now Equation~\eqref{eq:F1} implies that
\begin{equation} \label{eq:F-inter}
F(t)=  e^{-\theta_{1}} f( -\theta_{2} +  \delta x(t) ).
\end{equation}
Hence $x$ solves the ODE
\begin{equation*}
x'(t)= -\kappa x(t) -1 +   e^{-\theta_{1}} f( -\theta_{2} +  \delta x(t) ).
\end{equation*}
As one can see from \eqref{AtEqn}, this is exactly the ODE that $A$ satisfies. Since $A(0)=x(0)=0$, we obtain
\begin{eqnarray}\label{eq:XA}
A(t) \equiv x(t), \quad \text{for $t \in [0,T]$.}
\end{eqnarray}
Then \eqref{eq:F2} readily follows from \eqref{eq:F-inter}. In addition, we infer from \eqref{eq:XA} and \eqref{eq:xt} that
\begin{equation}\label{part1}
(\lambda_{0}-\mu) \cdot A(T) = (\lambda_{0}-\mu) \cdot  \int_0^T e^{- \kappa (T-s)}(F(s)-1)ds.
\end{equation}
Furthermore, Equation~\eqref{FA} implies that
\begin{equation}\label{part2}
\mu \int_0^T \kappa A(t)dt + \mu \cdot A(T) =  \int_{0}^{T}\mu(F(s)-1)ds.
\end{equation}
On combining \eqref{part1} and \eqref{part2}, we obtain \eqref{eq-siam}.
Therefore, we have recovered the result in \citet{Errais}.
\end{remark}

\begin{proof}[Proof of Theorem~\ref{MainThm}]
\citet{HawkesII} first discovered that a linear Hawkes process has an immigration-birth representation.
The immigrants (roots) arrive according to a time-inhomogeneous Poisson process $\bar{N}$ with intensity $\mu(t)\geq 0$ at time $t$.
Each immigrant generates children according to a Galton-Watson tree,
that is, the number of children of each immigrant follows a Poisson distribution
with parameter $\Vert h\Vert_{L^{1}}$, and each child will independently
generate children according to the same Poisson distribution, and so on and so forth.
In addition, when the children are born, they are born at independent random times
with the probability density function $\frac{h(t)}{\Vert h\Vert_{L^{1}}}$ for being born at time $t$.
In other words, they are born according to an inhomogeneous Poisson process with intensity $h(t)$.
Consider an immigrant arrive at time $0$.
Note that in the later computations, we will consider an immigrant that arrives
at a positive time $t$, but the computation is the same as shifting the time backwards
by $t$ to consider the immigrant that arrives at time $0$.
Let $K$ be the number
of children of this immigrant and $\ell_{1}$ be the associated jump size.
Let $S_{t, j}$ be the number
of the total descendants of the $j$-th child of the immigrant that were born before time $t$, including the $j$-th child,
and $L_{t, j}$ be the sum of all of jump sizes associated with
all the descendants of the $j$-th child, including $j$-th child, where $1\leq j\leq K$.
Let $S_{t}$ be the total number of all the descendants of this immigrant that arrives at time $0$
including the immigrant, and let $L_{t}^{S}$ be the associated sum of jump sizes,
that is $S_{t}=1+\sum_{j=1}^{K}S_{t, j}$ and $L_{t}^{S}=\ell_{1}+\sum_{j=1}^{K}L_{t, j}$.
Then, we have
\begin{align}
F(t)&:=\mathbb{E}\left[e^{-\theta_{1}S_{t}-\theta_{2}L_{t}^{S}}\right] \label{eq:F-def}
\\
&=\sum_{k=0}^{\infty}\mathbb{E}\left[e^{-\theta_{1}S_{t}-\theta_{2}L_{t}^{S}}|K=k\right]\mathbb{P}(K=k)
\nonumber
\\
&=\mathbb{E}\left[e^{-\theta_{1}-\theta_{2}\ell_{1}}
\sum_{k=0}^{\infty}\prod_{i=1}^{k}\mathbb{E}\left[e^{-\theta_{1}S_{t, i}-\theta_{2}L_{t, i}}\right]\mathbb{P}(K=k|\ell_{1})\right]
\nonumber
\\
&=\mathbb{E}\left[e^{-\theta_{1}-\theta_{2}\ell_{1}}
\sum_{k=0}^{\infty}\left(\mathbb{E}\left[e^{-\theta_{1}S_{t, 1}-\theta_{2}L_{t,1 } }\right]\right)^{k}\mathbb{P}(K=k|\ell_{1})\right]
\nonumber
\\
&=\mathbb{E}\left[e^{-\theta_{1}-\theta_{2}\ell_{1}}
\sum_{k=0}^{\infty}\left(\int_{0}^{t}\frac{h(s)}{\Vert h\Vert_{L^{1}}}F(t-s)ds\right)^{k}
e^{-\ell_{1}\Vert h\Vert_{L^{1}}}\frac{\ell_{1}^{k}\Vert h\Vert_{L^{1}}^{k}}{k!}\right]
\nonumber
\\
&=e^{-\theta_{1}}\mathbb{E}\left[e^{-\theta_{2}\ell_{1}+\int_{0}^{t}\ell_{1}h(s)(F(t-s)-1)ds}\right],
\nonumber
\end{align}
where the third and fourth equalities above use the tower property of the conditional expectation,
and the fact that $(S_{t, i},L_{t, i})$ are i.i.d. independent of $K$,
and the fifth equality above uses the fact that conditional on $\ell_{1}$,
$K$ is Poisson distributed with parameter $\ell_{1}\Vert h\Vert_{L^{1}}$
and conditional on the children being born at time $s$,
$e^{-\theta_{1}S_{t, 1}-\theta_{2}L_{t , 1}}$ has the expectation
$F(t-s)$ by the definition of $F(\cdot)$, and the timing
of the children being born at time $s$ has the probability density function $\frac{h(s)}{\Vert h\Vert_{L^{1}}}$.

Next, by the immigration-birth representation, we have
$N_{T}=\sum_{i:0<\bar{\tau}_{i}  \le T} S_{T-\bar \tau_{i}}(i)$,
and $L_{T}=\sum_{i:0<\bar{\tau}_{i}  \le T}L_{T- \bar \tau_{i}}^{S}(i)$,
where $\bar{\tau}_{i}$ are the arrival times of the time-inhomogeneous Poisson process $\bar{N}$
and $S_{T-t}(i)$ are i.i.d. copies of $S_{T-t}$,  and $L_{T-t}^{S}(i)$ are i.i.d. copies of $L_{T-t}^{S}$,
where $S_{T-t}$ and $L_{T-t}^{S}$ are defined as before.
Thus, we have
\begin{equation*}
\mathbb{E}\left[e^{-\theta_{1}N_{T}-\theta_{2}L_{T}}\right]
=\mathbb{E}\left[e^{\sum_{i:0<\bar{\tau}_{i} \le T}(-\theta_{1}S_{T-\bar \tau_{i}}(i)-\theta_{2}L_{T- \bar \tau_{i}}^{S}(i))}\right].
\end{equation*}
Hence, we have
\begin{align*}
\mathbb{E}\left[e^{-\theta_{1}N_{T}-\theta_{2}L_{T}}\right]
&=\mathbb{E}\left[\prod_{i:0<\bar{\tau}_{i}\leq T}e^{-\theta_{1}S_{T-\bar{\tau}_{i}}(i)-\theta_{2}L_{T-\bar{\tau}_{i}}^{S}(i)}
\right]
\\
&=\mathbb{E}\left[\mathbb{E}\left[\prod_{i:0<\bar{\tau}_{i}\leq T}e^{-\theta_{1}S_{T-\bar{\tau}_{i}}(i)-\theta_{2}L_{T-\bar{\tau}_{i}}^{S}(i)}\bigg|\mathcal{F}_{T}^{\bar{N}}\right]\right],
\end{align*}
where we used the tower property and $\mathcal{F}_{T}^{\bar{N}}$ is the natural filtration
generated by $\bar{N}$ process on the time interval $[0,T]$. 
Conditional on $\mathcal{F}_{T}^{\bar{N}}$, $(S_{T-\bar{\tau}_{i}},L_{T-\bar{\tau}_{i}})$
are independent. Thus, we have
\begin{align*}
\mathbb{E}\left[e^{-\theta_{1}N_{T}-\theta_{2}L_{T}}\right]
&=\mathbb{E}\left[\prod_{i:0<\bar{\tau}_{i}\leq T}
\mathbb{E}\left[e^{-\theta_{1}S_{T-\bar{\tau}_{i}}(i)-\theta_{2}L_{T-\bar{\tau}_{i}}^{S}(i)}
\bigg|\mathcal{F}_{T}^{\bar{N}}\right]\right]
\\
&=\mathbb{E}\left[e^{\sum_{i:0<\bar{\tau}_{i} \le T}
\log\mathbb{E}\left[e^{-\theta_{1}S_{T-\bar{\tau}_{i}}(i)-\theta_{2}L_{T-\bar{\tau}_{i}}^{S}(i)}\big|\mathcal{F}_{T}^{\bar{N}}\right]}\right]
\\
&=\mathbb{E}\left[e^{\sum_{i:0<\bar{\tau}_{i} \le T}
\log\mathbb{E}\left[e^{-\theta_{1}S_{T-\bar{\tau}_{i}}(1)-\theta_{2}L_{T-\bar{\tau}_{i}}^{S}(1)}\big|\mathcal{F}_{T}^{\bar{N}}\right]}\right]
\nonumber
\\
&=e^{\int_{0}^{T}\mu(s)\left(\mathbb{E}\left[e^{-\theta_{1}S_{T-s}(1)-\theta_{2}L_{T-s}^{S}(1)}\right]-1\right)ds}
\nonumber
\\
&=e^{\int_{0}^{T}\mu(s)(F(T-s)-1)ds},
\nonumber
\end{align*}
where the second last equality follows from
the fact that for any deterministic and bounded function $g(\cdot)$,
and the inhomogeneous Poisson process $\bar{N}$ with intensity $\mu(\cdot)$,
we have $\mathbb{E}[e^{\int_{0}^{T}g(s)d\bar{N}_{s}}]=e^{\int_{0}^{T}\mu(s)(e^{g(s)}-1)ds}$,
and the last equality follows from the definition of $F$ in \eqref{eq:F-def}.

Finally, we show that $F$ is the unique solution to the integral equation \eqref{eq:F} satisfying $|F(t)| \le 1$ for all $t \in [0, T].$ The fact that $|F(t)| \le 1$ is clear from the definition \eqref{eq:F-def}.
To show uniqueness, let $F_{1}(t)$ and $F_{2}(t)$ be two solutions of \eqref{eq:F}
so that $|F_{1}(t)|,|F_{2}(t)|\leq 1$ for $t \in [0,T].$
Then we have
\begin{align*}
|F_{1}(t)-F_{2}(t)|
&\leq\mathbb{E}\left[\left|e^{-\theta_{1}}e^{-\theta_{2}\ell_{1}+\int_{0}^{t}\ell_{1}h(s)(F_{1}(t-s)-1)ds}
-e^{-\theta_{1}}e^{-\theta_{2}\ell_{1}+\int_{0}^{t}\ell_{1}h(s)(F_{2}(t-s)-1)ds}\right|\right]
\\
&=\mathbb{E}\left[\left|e^{-\theta_{1}}e^{-\theta_{2}\ell_{1}}\right| \cdot
\left|e^{\int_{0}^{t}\ell_{1}h(s)(F_{1}(t-s)-1)ds}
-e^{\int_{0}^{t}\ell_{1}h(s)(F_{2}(t-s)-1)ds}\right|\right]
\\
&\leq\mathbb{E}\left[\left|e^{\int_{0}^{t}\ell_{1}h(s)(F_{1}(t-s)-1)ds}
-e^{\int_{0}^{t}\ell_{1}h(s)(F_{2}(t-s)-1)ds}\right|\right].
\end{align*}
Let $F_{1}(t)=R_{1}(t)+iI_{1}(t)$ and $F_{2}(t)=R_{2}(t)+iI_{2}(t)$.
Then, we have
\begin{align}\label{FEqn}
&|F_{1}(t)-F_{2}(t)|
\\
&\leq\mathbb{E}\left[\left|e^{\int_{0}^{t}\ell_{1}h(s)(R_{1}(t-s)-1)ds+i\int_{0}^{t}\ell_{1}h(s)I_{1}(t-s)ds}
-e^{\int_{0}^{t}\ell_{1}h(s)(R_{2}(t-s)-1)ds+i\int_{0}^{t}\ell_{1}h(s)I_{1}(t-s)ds}\right|\right]
\nonumber
\\
&\qquad
+\mathbb{E}\left[\left|e^{\int_{0}^{t}\ell_{1}h(s)(R_{2}(t-s)-1)ds+i\int_{0}^{t}\ell_{1}h(s)I_{1}(t-s)ds}
-e^{\int_{0}^{t}\ell_{1}h(s)(R_{2}(t-s)-1)ds+i\int_{0}^{t}\ell_{1}h(s)I_{2}(t-s)ds}\right|\right]
\nonumber
\\
&=\mathbb{E}\left[\left|e^{\int_{0}^{t}\ell_{1}h(s)(R_{1}(t-s)-1)ds}
-e^{\int_{0}^{t}\ell_{1}h(s)(R_{2}(t-s)-1)ds}\right|\right]
\nonumber
\\
&\qquad
+\mathbb{E}\left[e^{\int_{0}^{t}\ell_{1}h(s)(R_{2}(t-s)-1)ds}\left|e^{i\int_{0}^{t}\ell_{1}h(s)I_{1}(t-s)ds}
-e^{i\int_{0}^{t}\ell_{1}h(s)I_{2}(t-s)ds}\right|\right].
\nonumber
\end{align}
Notice that $|R_{1}(t)|\leq |F_{1}(t)|\leq 1$ and $|R_{2}(t)|\leq |F_{2}(t)|\leq 1$,
thus, $\int_{0}^{t}\ell_{1}h(s)(R_{j}(t-s)-1)ds\leq 0$
for $j=1,2$. The map $x\mapsto e^{x}$ is Lipschitz with constant $1$ for $x\leq 0$.
Thus, for any $0\leq t\leq T$.
\begin{align}
\mathbb{E}\left[\left|e^{\int_{0}^{t}\ell_{1}h(s)(R_{1}(t-s)-1)ds}
-e^{\int_{0}^{t}\ell_{1}h(s)(R_{2}(t-s)-1)ds}\right|\right]
&\leq\mathbb{E}\left[\int_{0}^{t}\ell_{1}h(s)|R_{1}(t-s)-R_{2}(t-s)|ds\right]
\nonumber
\\
&\leq
\Vert h\Vert_{L^{\infty}[0,T]}\mathbb{E}[\ell_{1}]\int_{0}^{t}|R_{1}(s)-R_{2}(s)|ds
\nonumber
\\
&\leq
\Vert h\Vert_{L^{\infty}[0,T]}\mathbb{E}[\ell_{1}]\int_{0}^{t}|F_{1}(s)-F_{2}(s)|ds,
\label{REqn}
\end{align}
where $\Vert h\Vert_{L^{\infty}[0,T]}=\sup_{0\leq s\leq T}h(s)$.

Next, let us notice that for any $x,y\in\mathbb{R}$,
\begin{equation*}
|e^{ix}-e^{iy}|
\leq|\cos(x)-\cos(y)|+|\sin(x)-\sin(y)|
\leq 2|x-y|.
\end{equation*}
Therefore,
\begin{align}
&\mathbb{E}\left[e^{\int_{0}^{t}\ell_{1}h(s)(R_{2}(t-s)-1)ds}\left|e^{i\int_{0}^{t}\ell_{1}h(s)I_{1}(t-s)ds}
-e^{i\int_{0}^{t}\ell_{1}h(s)I_{2}(t-s)ds}\right|\right]
\nonumber
\\
&\leq
\mathbb{E}\left[\left|e^{i\int_{0}^{t}\ell_{1}h(s)I_{1}(t-s)ds}
-e^{i\int_{0}^{t}\ell_{1}h(s)I_{2}(t-s)ds}\right|\right]
\nonumber
\\
&\leq 2\mathbb{E}\left[\int_{0}^{t}\ell_{1}h(s)|I_{1}(t-s)-I_{2}(t-s)|ds\right]
\nonumber
\\
&\leq
2\Vert h\Vert_{L^{\infty}[0,T]}\mathbb{E}[\ell_{1}]\int_{0}^{t}|I_{1}(s)-I_{2}(s)|ds
\nonumber
\\
&\leq
2\Vert h\Vert_{L^{\infty}[0,T]}\mathbb{E}[\ell_{1}]\int_{0}^{t}|F_{1}(s)-F_{2}(s)|ds.
\label{IEqn}
\end{align}
Hence, by applying \eqref{REqn} and \eqref{IEqn} to \eqref{FEqn}, we get
\begin{equation*}
|F_{1}(t)-F_{2}(t)|
\leq
3\Vert h\Vert_{L^{\infty}[0,T]}\mathbb{E}[\ell_{1}]\int_{0}^{t}|F_{1}(s)-F_{2}(s)|ds.
\end{equation*}
By Gronwall's inequality, we conclude that $F_{1}\equiv F_{2}$ on $[0,T]$. The proof is complete.
\end{proof}

By letting $\theta_{1}=0$ or $\theta_{2}=0$ in Theorem~\ref{MainThm}, we get the single Laplace transforms
of the counting process $N_{T}$ and the point process $L_{T}$.

\begin{corollary}\label{SingleTransform}
(i) For any $\theta\in\mathbb{C}$ with $\mathcal{R}(\theta)\geq 0$,
\begin{equation*}
\mathbb{E}\left[e^{-\theta N_{T}}\right]
=e^{\int_{0}^{T}\mu(T-s)(F_{N}(s)-1)ds},
\end{equation*}
where the function $F_{N}$ is the unique solution to the integral equation
\begin{equation}\label{FN}
F_{N}(t)=e^{-\theta}\mathbb{E}\left[e^{\int_{0}^{t}\ell_{1}h(s)(F_{N}(t-s)-1)ds}\right],
\end{equation}
with $|F_{N}(t)|\leq 1$ for $0\leq t\leq T$.

(ii) For any $\theta\in\mathbb{C}$ with $\mathcal{R}(\theta)\geq 0$,
\begin{equation} \label{eq:L-lap}
\mathbb{E}\left[e^{-\theta L_{T}}\right]
=e^{\int_{0}^{T}\mu(T-s)(F_{L}(s)-1)ds},
\end{equation}
where the function $F_{L}$ is the unique solution to the integral equation
\begin{equation}\label{FL}
F_{L}(t)=\mathbb{E}\left[e^{-\theta\ell_{1}+\int_{0}^{t}\ell_{1}h(s)(F_{L}(t-s)-1)ds}\right],
\end{equation}
with $|F_{L}(t)|\leq 1$ for $0\leq t\leq T$.
\end{corollary}

\begin{remark}
The result of the single Laplace transform of $N_{T}$ has been obtained in \citet{Karabash}
by using the immigration-birth representation as a special case of the linear marked Hawkes process.
\end{remark}

The Laplace transforms obtained allow us to explicitly compute the moments of the counting process $N_{T}$ and the point process $L_{T}$. We derive the first and second moments in the following result and present the proof in the Appendix.
Higher order moments can be derived similarly.

\begin{proposition} \label{prop:moments}
(i) The first moment of the counting process $N$ is given by
\begin{equation*}\label{eq:E-NT}
\mathbb{E}[N_{T}]=\int_{0}^{T}\mu(T-t)\psi_{1}(t)dt,
\end{equation*}
where $\psi_{1}$ is the unique solution to the equation:
\begin{equation}\label{psi1}
\psi_{1}(t)
=1+\int_{0}^{t}\mathbb{E}[\ell_{1}]h(t-s)\psi_{1}(s)ds,
\qquad
0\leq t\leq T.
\end{equation}

(ii) The first moment of the process $L$ is given by
\begin{equation*} \label{eq:ELT}
\mathbb{E}[L_{T}]=\mathbb{E}[\ell_{1}]\int_{0}^{T}\mu(T-t)\psi_{1}(t)dt.
\end{equation*}

(iii) The second moment of the counting process $N$ is given by
\begin{equation*}
\mathbb{E}[N_{T}^{2}]=\int_{0}^{T}\mu(T-t)\psi_{2}(t)dt+\left(\int_{0}^{T}\mu(T-s)\psi_{1}(t)dt\right)^{2},
\end{equation*}
where $\psi_{2}$ is the unique solution
to the equation:
\begin{equation}\label{psi2}
\psi_{2}(t)
=(\psi_{1}(t))^{2}
+\int_{0}^{t}\mathbb{E}[\ell_{1}]h(s)\psi_{2}(t-s)ds,
\qquad
0\leq t\leq T.
\end{equation}

(iv) The second moment of the process $L$ is given by
\begin{equation*}
\mathbb{E}[L_{T}^{2}]
=\int_{0}^{T}\mu(T-t)\psi_{3}(t)dt+(\mathbb{E}[\ell_{1}])^{2}\left(\int_{0}^{T}\mu(T-t)\psi_{1}(t)dt\right)^{2},
\end{equation*}
where $\psi_{1}$ is defined in \eqref{psi1} and $\psi_{3}$ is the unique solution
to the equation:
\begin{equation}\label{psi3}
\psi_{3}(t)
=\mathbb{E}[\ell_{1}^{2}](\psi_{1}(t))^{2}
+\int_{0}^{t}\mathbb{E}[\ell_{1}]h(s)\psi_{3}(t-s)ds,
\qquad
0\leq t\leq T.
\end{equation}
\end{proposition}
\begin{remark}
It follows from Proposition~\ref{prop:moments}
that the first moments $\mathbb{E}[N_{T}]$ and $\mathbb{E}[L_{T}]$
and also the second moment $\mathbb{E}[N_{T}^{2}]$
depend on the distribution of $\ell_{1}$ only via the mean $\mathbb{E}[\ell_{1}]$,
and the second moment $\mathbb{E}[L_{T}^{2}]$
depends on the distribution of $\ell_{1}$ only via the $\mathbb{E}[\ell_{1}]$, $\mathbb{E}[\ell_{1}^{2}]$,
the first two moments of $\ell_{1}$.
\end{remark}

Using the Laplace transform of $N_{T}$, one can also compute the probability
mass function of $N_{T}$ analytically, as shown in the following result.
The proof relies on the celebrated Fa\`{a} di Bruno's formula, and it is given in the appendix.

\begin{proposition}\label{prop:mass}
We have $\mathbb{P}(N_{T}=0)=e^{-\int_{0}^{T}\mu(T-s)ds}$,
and for any $k\geq 1$,
\begin{align*}
\mathbb{P}(N_{T}=k)
&=
e^{-\int_{0}^{T}\mu(T-s)ds}\sum\frac{1}{m_{1}!1!^{m_{1}}m_{2}!2!^{m_{2}}\cdots m_{k}!k!^{m_{k}}}
\\
&\qquad\qquad\qquad
\cdot\prod_{j=1}^{k}\left(\int_{0}^{T}\mu(T-s)F_{N,j}(s)ds\right)^{m_{j}},
\end{align*}
where the summation is over all $k$-tuples of nonnegative integers $(m_1, \ldots, m_k)$ satisfying the constraint $1\cdot m_{1}+2\cdot m_{2}+3\cdot m_{3}+\cdots+k\cdot m_{k}=k$,
and $F_{N,0}(t)=0$,
\begin{equation*}
F_{N,1}(t)=\mathbb{E}\left[e^{-\int_{0}^{t}\ell_{1}h(s)ds}\right],
\end{equation*}
and for every $j\geq 2$,
\begin{align*}
F_{N,j}(t)
&=\sum\frac{j!}{m_{1}!1!^{m_{1}}m_{2}!2!^{m_{2}}\cdots m_{j-1}!(j-1)!^{m_{j-1}}}
\\
&\qquad\qquad
\cdot
\mathbb{E}\left[e^{-\int_{0}^{T}\ell_{1}h(s)ds}\prod_{i=1}^{j-1}\left(\int_{0}^{t}\ell_{1}h(s)F_{N,i}(t-s)ds\right)^{m_{i}}\right],
\end{align*}
where the summation is over all $(j-1)$-tuples of nonnegative integers $(m_1, \ldots, m_{j-1})$ satisfying the constraint $1\cdot m_{1}+2\cdot m_{2}+3\cdot m_{3}+\cdots+(j-1)\cdot m_{j-1}=j-1$.
\end{proposition}


Next, let us discuss the distribution of $L_{T}$. First note that by assuming that $\mathbb{P}(\ell_{1}=0)=0$,
we always have $\mathbb{P}(L_{T}=0)=\mathbb{P}(N_{T}=0)=e^{-\int_{0}^{T}\mu(s)ds}$ for any nonnegative jump size distribution of $\ell_1$.  Next, we assume that the random jump size $\ell_{1}$ has a lattice distribution and
takes discrete values $k\delta$, $k\in\mathbb{N}$
with $\mathbb{P}(\ell_{1}=k\delta)=p_{k}$
where $0\leq p_{k}\leq 1$ and $\sum_{k=1}^{\infty}p_{k}=1$ for some $\delta>0.$
Note that this includes the case for geometrically distributed
$\ell_{1}$, Poisson distributed $\ell_{1}$ etc. for fixed $\delta=1$.
Under this assumption, $L_{T}$ also takes values
$k\delta$, for $k\in\mathbb{N}\cup\{0\}$.

Then we have the following result on the distribution of $L_{T}$ when the jump size $\{\ell_i\}$ is lattice distributed. The proof is deferred to the Appendix.
\begin{proposition} \label{prop:LT-distri}
We have $\mathbb{P}(L_{T}=0)=e^{-\int_{0}^{T}\mu(T-s)ds}$,
and for any $k\geq 1$,
\begin{align*}
\mathbb{P}(L_{T}=k\delta)
&=
e^{-\int_{0}^{T}\mu(T-s)ds}\sum\frac{1}{m_{1}!1!^{m_{1}}m_{2}!2!^{m_{2}}\cdots m_{k}!k!^{m_{k}}}
\\
&\qquad\qquad\qquad
\cdot\prod_{j=1}^{k}\left(\int_{0}^{T}\mu(T-s)F_{L,j}(s)ds\right)^{m_{j}},
\end{align*}
where the summation is over all $k$-tuples of nonnegative integers $(m_1, \ldots, m_k)$ satisfying the constraint $1\cdot m_{1}+2\cdot m_{2}+3\cdot m_{3}+\cdots+k\cdot m_{k}=k$,
and $F_{L,0}(t)=0$,
\begin{equation*}
F_{L,1}(t)=p_{1}e^{-\int_{0}^{t}\delta h(s)ds},
\end{equation*}
and for every $j\geq 2$,
\begin{align*}
F_{L,j}(t)&=\sum_{k=0}^{j}\binom{j}{k}k!
e^{-\int_{0}^{t}k\delta h(s)ds}p_{k}
\sum\frac{(j-k)!}{m_{1}!1!^{m_{1}}m_{2}!2|^{m_{2}}\cdots
m_{j-k}!(j-k)!^{m_{j-k}}}
\\
&\qquad\qquad\qquad\cdot
\prod_{i=1}^{j-k}\left(\int_{0}^{t}k\delta h(s)F_{L,i}(t-s)ds\right)^{m_{i}},
\end{align*}
where the summation is over all $(j-k)$-tuples of nonnegative integers $(m_1, \ldots, m_{j-k})$ satisfying the constraint $1\cdot m_{1}+2\cdot m_{2}+\cdots+(j-k)\cdot m_{j-k}=j-k$.
\end{proposition}

Numerical methods to calculate
the summands in Fa\`{a} di Bruno's formula in Propositions~\ref{prop:mass} and \ref{prop:LT-distri} can be found in, e.g. \citet{Kilmko73}. In general, when the jump size $\ell_{1}$ is not lattice distributed, one can still efficiently calculate the distributions of $N_T$ and $L_T$ by numerically inverting the Laplace transforms in Corollary~\ref{SingleTransform}. See e.g. \citet{AbateWhitt95} for numerical Laplace transform inversion methods.

\section{Applications in Dark Pool Trading} \label{sec:application}

In this section, we apply the main results to analyze performance problems arising from trading in dark pools. We use the Hawkes process to model executions of a large midpoint peg order placed at an empty dark pool and compute various performances in Section~\ref{sec:performance}.  Non-empty dark pools are discussed in Sections~\ref{sec:nonempty}. In computing some performance metrics (e.g. the probability of obtaining another fill conditioned on a past fill) and studying non-empty pools, we will see that it is natural to study a Hawkes process with a time-dependent baseline intensity.

\subsection{Model description and performance analysis}\label{sec:performance}

Suppose an investor rests a large midpoint buy order of size $x>0$ in a dark pool with a continuous first-come-first-served matching mechanism. This order is ``pegged'' at the mid-quote of transparent exchanges, i.e., the execution price of the order automatically adjusts as the market moves. Considering a sell order is similar.
As liquidity in dark pools could be sparse and there could be a high probability of no volume in pools (see, e.g., \citet{Ganchev2010, Markov2013}), we assume in this section that when the investor's order reaches the dark pool there are no other orders sitting in the pool.

We model the successive executions of this midpoint peg order using a Hawkes process. More specifically,
we model the consolidated sell trades from other players in the dark pool as a Hawkes process $(N, L)$ with the intensity \eqref{eq:dynamics}, where $N_t$ counts the number of eligible-to-trade sell orders (or trades with the investor's resting buy order) by time $t$ and the i.i.d. sequence $\{\ell_i: i =1, 2, \ldots\}$ models the volumes of arriving sell orders. Such a self-exciting Hawkes process based model of executions of a large order could capture the clustering of trade arrivals and positive liquidity feedback in dark pools.

Since the pool is assumed to be initially empty, there will be no trade occurring at time zero when the investor puts the buy order in the pool. 
A sample path of the trading intensity $\lambda_t$ and the remaining quantity of the dark order is given in Figure~\ref{fig1}.

\begin{figure}[htb]
	\centering
	\subfigure[Trading intensity $\lambda_t$]{
		\begin{minipage}[b]{0.45\textwidth}
			\centering
			\includegraphics[width=\textwidth]{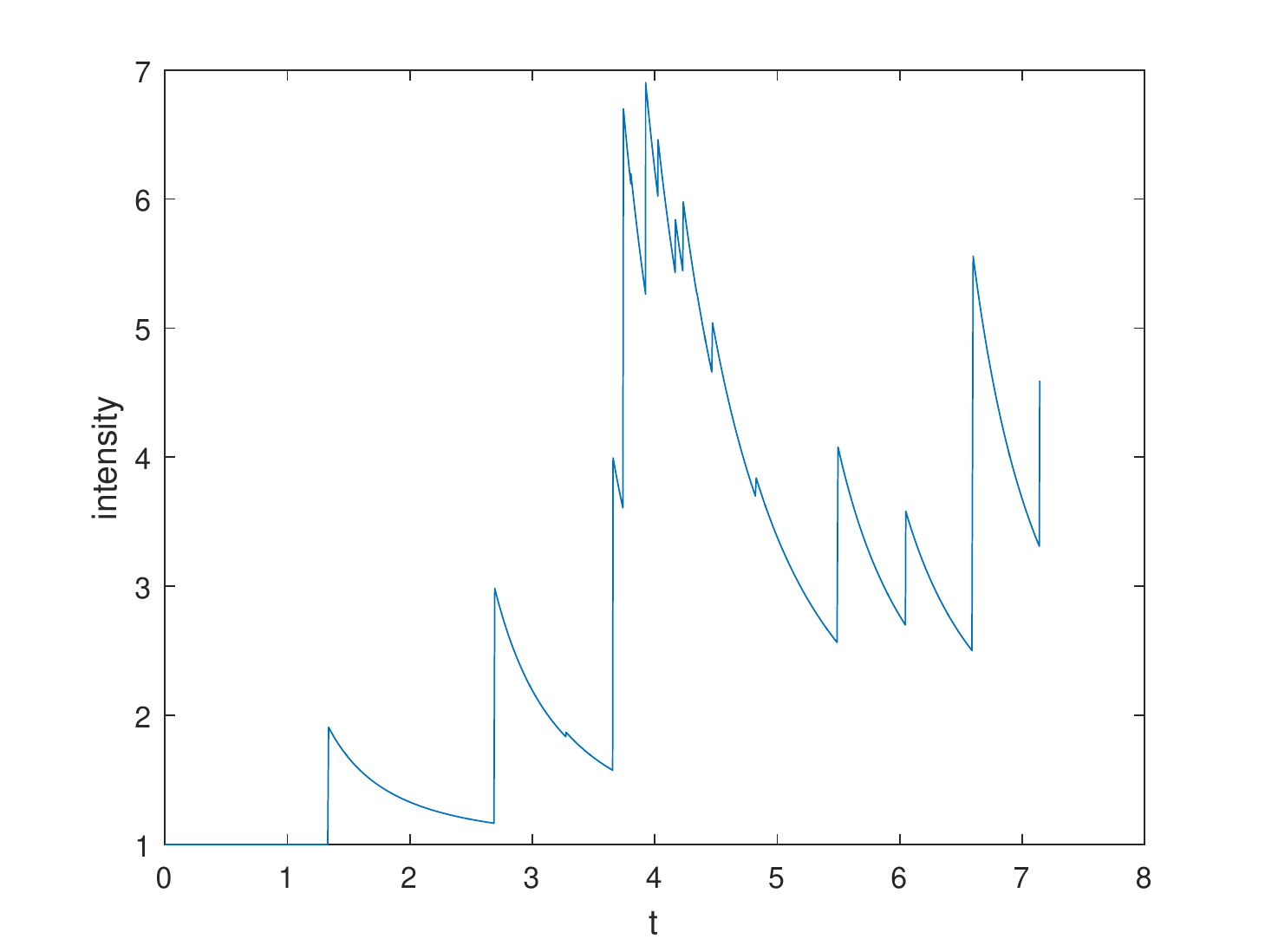}
	\end{minipage}	}			
	\subfigure[Remaining order size]{
		\label{fig:mini:subfig:1}
		\begin{minipage}[b]{0.45\textwidth}
			\centering
			\includegraphics[width=\textwidth]{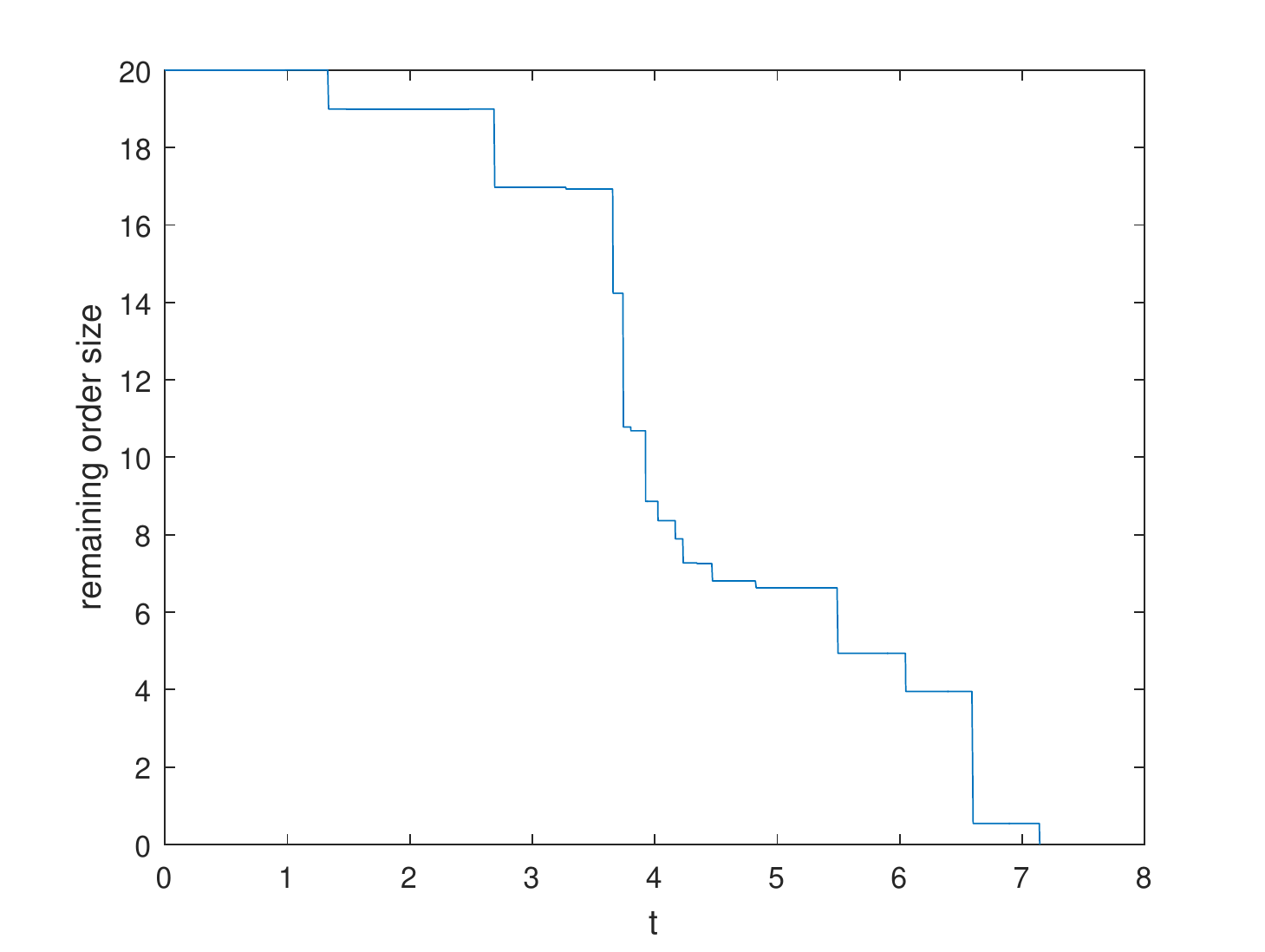}
	\end{minipage}}	
	\caption{\emph{(a) A sample path of the intensity $\lambda_t$ of a Hawkes process and (b) a sample path of the remaining order size $(x-L_t)^+$ for a resting dark order with initial size $x=20$. Here for the Hawkes model \eqref{eq:dynamics}, the baseline intensity $\mu(t)\equiv 1$, the trade size $\ell_i$ follows an exponential distribution with mean one, and the exciting function $h(t)= \frac{9}{10} \frac{1}{(1+t)^2}$. Note that this Hawkes model is non-Markovian.}}
	\label{fig1}
\end{figure}
For this particular path, we observe from Figure~\ref{fig:mini:subfig:1} that the resting order matches with incoming sell trades with variable sizes, and it will be completely filled at time $t=7.2$ if the investor leaves the order in the pool
for a sufficiently long time. On the other hand, if the investor decides to cancel the order before time $t=7.2$, then this resting order will be partially executed and the remaining quantity could be routed to another dark pool or a lit exchange for liquidity-seeking purposes.

We want to compute various performance quantities of interest such as time-to-first-fill and fill rate which indicate the liquidity of a dark venue. Their mathematical expressions and economic interpretations are summarized below. In a fragmented financial market with multiple dark pools and exchanges, these performance metrics could be useful for smart order routing and allocation of liquidity among different pools to maximize fills and liquidity opportunities from dark pools, which in turn help investors reduce market impact or opportunity cost in trading big orders.
In terms of the notations, we differentiate between $\ell_{i}$ and $l_{i}$ by
having $\ell_{i}$ being random and $l_{i}$ being deterministic and given.

\textbf{Performance quantities we consider}:
 \begin{enumerate}
\item [(a)] Time-to-first-fill $\tau(1)$ of the order is defined by
\begin{equation*}
\tau(1):=\inf\{t \ge 0: N_t =1\}.
\end{equation*}
That is, $\tau(1)$ measures the time between order placement at a given dark venue and the first execution (possibly a partial fill) of that order. Thus we obtain that the probability of a fill within $[0,t]$ is given by
\begin{equation*}
\mathbb{P}(\tau({1}) \le t)=1 - \mathbb{P}(N_t=0)=1 - e^{-\int_{0}^{t}\mu(s)ds}, \quad \text{for $t \ge 0$}.
\end{equation*}
>From this expression, it is clear that the baseline intensity $\mu(\cdot)$ completely determines the distribution of time-to-first-fill.
In particular, when $\mu(\cdot)\equiv\mu$ is constant, $\tau(1)$ is an exponential random variable with mean $1/\mu$.

 \item [(b)] Time-to-complete-fill $\sigma_x$ of a resting order with size $x>0$ is defined by the time it takes for the order to be completely executed. That is, $\sigma_x$ measures the time it takes
 for the aggregated volumes of matching trades exceed the resting order's size $x$:
 \begin{equation} \label{eq:TTF}
 \sigma_x:=\inf\{t \ge 0: L_t\ge x \}.
 \end{equation}
 Hence its distribution is given by
 \begin{equation*}
\mathbb{P}(\sigma_{x}\leq t)=\mathbb{P}(L_t\geq x), \quad \text{for $t \ge 0$}.
\end{equation*}
Since we have obtained the Laplace transform of $L_t$ in \eqref{eq:L-lap}, we can then use the inverse Laplace transform to obtain the distribution of $L_t$ and that of $\sigma_x$ numerically. In addition, we can also compute the expected time-to-complete-fill $\mathbb{E}[\sigma_x]$ numerically, where
\begin{equation} \label{eq:E-TTF}
\mathbb{E}[\sigma_x] = \int_{0}^{\infty} \mathbb{P}(\sigma_{x}> t) dt = \int_{0}^{\infty} \mathbb{P}(L_t < x) dt.
\end{equation}

 \item [(c)] The expected fill rate of the resting dark order with size $x>0$ in the time interval $[0, t]$ is defined by
 \begin{equation}\label{eq:fillrate}
 \frac{1}{x} \cdot \mathbb{E} [\min\{L_{t},x\}],
 \end{equation}
 which is equal to $\frac{1}{x}\int_{0}^x \mathbb{P}(L_t >y)dy$.
 In practice, the deadline $t$ can be deterministic or random. For example, the investor may rest the order in a particular dark pool for one minute which is predetermined at the time of order placement. It is also possible that the investor may cancel a resting order due to exogenous market events such as a significant price move, in which case $t$ is random. We can numerically evaluate the expected fill rate \eqref{eq:fillrate} efficiently if $t$ is independent of the execution process $(N, L)$ by first inverting the Laplace transform of $L_t$ in \eqref{eq:L-lap} and getting its distribution, and then calculate the expectation in \eqref{eq:fillrate}. Alternatively, one can also use Fast Fourier Transform (FFT) methods where the expected fill rates \eqref{eq:fillrate} across the whole spectrum of order sizes $x$ can be obtained in one set of FFT calculations. See e.g. \cite{CarrMadan}.

 \item [(d)] The probability of obtaining one fill (or at least one fill) in the next $T$ units of time, given that there is an initial fill of size $l_1<x$ in $(0,t]$. Mathematically we are interested in computing
\begin{equation}\label{eq:cond-hit}
\mathbb{P}(N_{t+T}-N_t=1|N_t=1,\ell_{1}=l_{1}),
\end{equation}
and
\begin{equation}\label{eq:cond-hit2}
\mathbb{P}(N_{t+T}-N_t\ge 1|N_t=1,\ell_{1}=l_{1}) = 1 - \mathbb{P}(N_{t+T}-N_t=0|N_t=1,\ell_{1}=l_{1}).
\end{equation}
As argued in the industry paper \citet{Mittal2007}, these conditional fill probabilities are particularly interesting in practice. Liquidity in a dark pool is sticky, and the expectation of liquidity changes when a trade occurs. The conditional fill probabilities in \eqref{eq:cond-hit} and \eqref{eq:cond-hit2} give investors a quantitative view of the liquidity expectation in the future given a prior fill of the resting order.

To compute these conditional fill probabilities, we can use the Laplace transform of $N_T$ in Corollary~\ref{SingleTransform} and the intensity dynamics \eqref{eq:dynamics} to obtain that
\begin{align} \label{eq:cond-prob}
&\mathbb{P}(N_{t+T} - N_t=0|N_t=1,\ell_{1}=l_{1}) \nonumber
\\
&=\frac{\int_{0}^{t}\mu(\tau_{1})e^{-\int_{0}^{\tau_{1}}\mu(s)ds}e^{-\int_{\tau_{1}}^{t}(\mu(s)+h(s-\tau_{1})l_{1})ds}
e^{-\int_{t}^{t+T}(\mu(s)+h(s-\tau_{1})l_{1})ds}d\tau_{1}}
{\int_{0}^{t}\mu(\tau_{1})e^{-\int_{0}^{\tau_{1}}\mu(s)ds}e^{-\int_{\tau_{1}}^{t}(\mu(s)+h(s-\tau_{1})l_{1})ds}d\tau_{1}},
\end{align}
and
\begin{align} \label{eq:cond-prob2}
&\mathbb{P}( N_{t+T} - N_t=1|N_t=1,\ell_{1}=l_{1}) \nonumber
\\
&=\left(\int_{0}^{t}\mu(\tau_{1})e^{-\int_{0}^{\tau_{1}}\mu(s)ds}e^{-\int_{\tau_{1}}^{t}(\mu(s)+h(s-\tau_{1})l_{1})ds}d\tau_{1}\right)^{-1}
\nonumber \\
&\quad \cdot\int_{0}^{t}\bigg[\mu(\tau_{1})e^{-\int_{0}^{\tau_{1}}\mu(s)ds}e^{-\int_{\tau_{1}}^{t+T}(\mu(s)+h(s-\tau_{1})l_{1})ds}
\nonumber
\\
&\qquad\qquad\qquad
\cdot
\int_{t}^{t+T}(\mu(s)+h(s-\tau_{1})l_{1})\mathbb{E} \left[e^{-\int_{0}^{T+t-s} \ell_1 h(u)du}\right]ds\bigg]d\tau_{1}.
\end{align}

A detailed derivation of \eqref{eq:cond-prob} and \eqref{eq:cond-prob2} relies on the distributional properties of a Hawkes process with a time-dependent baseline intensity (due to the conditioning), and it
is given in the Appendix.

Two interesting observations are in order.
First, we can infer from \eqref{eq:cond-hit2} and \eqref{eq:cond-prob} that the conditional probability of at least one fill given there is a past fill of size $l_1$ in the last $t$ units of time, is independent of the distribution of the trade size $\ell_1.$ Second, we note from \eqref{eq:cond-prob2} that the conditional probability of exactly one fill in the next $T$ units of time depends on the distribution of $\ell_1$ only through its Laplace transform.

\item [(e)] The expected fill size of the resting dark order in the next $T$ units of time conditioned on there is an initial fill of size $l_1<x$ in $(0,t]$, is given by
\begin{equation}\label{eq:conditional-expec}
\mathbb{E}\left[(L_{t+T} - L_t) \wedge (x-l_1) |N_t = 1, \ell_1 =l_1\right]
=\mathbb{E}\left[\min\{L_{t+T} , x\} |N_t = 1, \ell_1 =l_1\right] - l_1.
\end{equation}
Similar as the conditional fill probabilities, such a conditional expectation provides the investor with an indication of the liquidity size in the dark pool based on a prior execution of the dark order.

To compute this conditional expected fill size, we can first infer from Corollary~\ref{SingleTransform} and the intensity dynamics \eqref{eq:dynamics} to obtain the following Laplace transform:
\begin{align}
&\mathbb{E}\left[e^{-\theta (L_{t+T} - L_t) }|N_t=1,\ell_{1}=l_{1}\right] \nonumber
\\
&=\frac{\int_{0}^{t}\mu(\tau_{1})e^{-\int_{0}^{\tau_{1}}\mu(s)ds}e^{-\int_{\tau_{1}}^{t}(\mu(s)+h(s-\tau_{1})l_{1})ds}
e^{\int_{t}^{t+T}(\mu(s)+h(s-\tau_{1})l_{1})(F_{L}(T+t-s)-1)ds}d\tau_{1}}
{\int_{0}^{t}\mu(\tau_{1})e^{-\int_{0}^{\tau_{1}}\mu(s)ds}e^{-\int_{\tau_{1}}^{t}(\mu(s)+h(s-\tau_{1})l_{1})ds}d\tau_{1}},
\label{pgfL}
\end{align}
where for any $0\leq t\leq T$, the function $F_L$ satisfies the integral equation:
\begin{equation*}
F_{L}(t)=\mathbb{E}\left[e^{-\theta\ell_{1}+\int_{0}^{t}\ell_{1}h(s)(F_{L}(t-s)-1)ds}\right].
\end{equation*}
Then, we can numerical invert this Laplace transform to obtain the conditional distribution of $L_{t+T} - L_t$ and hence the conditional expectation in \eqref{eq:conditional-expec}.
The derivation of \eqref{pgfL} is similar as the derivation of \eqref{pgfN} in the Appendix, where we use the properties of a Hawkes process with a time-dependent baseline intensity. We omit the derivation here.
\end{enumerate}

Two remarks are in order. First, the estimations of the performance metrics of a resting order are relatively straightforward if the investor has his/her own execution data from trading in dark pools. For example, the expected fill rate of an order of size $x$ placed at a dark pool for a given time horizon can be estimated as the arithmetic average of the fill rate of many orders with the same sizes $x$ placed at this pool, assuming that the market conditions and pool characteristics remain stationary. The estimation procedure is similar for other performance quantities. Second, given the investor has the execution data from trading in a dark pool, it is also possible to estimate the Hawkes model by first estimating the trade size distribution $\ell_i$, and then estimate the baseline intensity and the exciting function using parametric or non--parametric methods. See e.g. \cite{Bacry2015, Errais} and references therein for details on estimations of Hawkes models.


\subsubsection*{Numerical examples}
We now present numerical examples to illustrate the computations of the various quantities derived in Section~\ref{sec:performance}. We also consider different order size distributions for $\ell_i$, different exciting functions $h(\cdot)$, and different baseline intensities $\mu(\cdot)$
to investigate their impact on the performance metrics. Our numerical experiments are implemented in MATLAB on a PC with a 3.30 GHz Intel Processor and 8 GB of RAM.

To compute the performance quantities, we need to get the distribution of $L_t$ numerically. This requires us first to solve the integral equation \eqref{FL} to obtain the point process transform $F_L(t) $ for a given $\theta \in\mathbb{C}$ with a nonnegative real part, and then use Laplace inversion methods to obtain the distribution of $L_t$ for fixed $t$.  The Laplace inversion method we use is a Fourier
series method which employs Bromwich contour inversion integral and Euler summation. See \citet{AbateWhitt95} for a detailed
description of this Laplace inversion algorithm (called EULER in the paper).

To numerically solve $F_L(t)$ which satisfies a Hammerstein--type Volterra integral equation as in \eqref{FL}, we apply the collocation method, see e.g. Chapter~2.3.3 in \citet{Brunner2004}.
The main idea of this method is to select a number of points (collocation points) on $[0, t]$, and use piecewise polynomial functions to approximate the true solution where the piecewise polynomial functions solve the given integral equation at the collocation points. Table~\ref{table:time} reports the computation time for representative examples where we solve $F_L(t)$ for $ t \in [0,6]$ using piecewise linear approximation on $[0, 6]$ with a uniform mesh consisting of 150 subintervals, a number which balances the speed and accuracy of the algorithm. The computation time is generally around 25 seconds for various different specifications of the mark size $\ell_i$ and the exciting function $h(\cdot)$.

\begin{table}[h]
		\centering
		\begin{tabular}{cccc}
			\hline	
			& $\ell_i \equiv 1$ & $\ell_i \sim$ exponential & $\ell_i \sim$ hyper exponential  \\
			\hline
			$h_1(t)$ &25.411 &24.795  &25.343 \\
			
			$h_2(t)$ &25.349 &25.048   &25.472\\
			
			$h_3(t)$ & 25.433&  25.778 &25.801\\
			\hline
		\end{tabular}
		\caption{For a given $\theta \in\mathbb{C}$ ($\theta =1+ i$ in this example), this table records the CPU time (in seconds) for using piecewise linear collocation method to solve $F_L(t)$ on the time interval $[0,6]$ with a uniform mesh consisting of 150 subintervals for different combinations of the mark size $\ell_i$ and the exciting function $h(t)$. Here, we have considered three distributions for $\ell_i$: (a) $\ell_i \equiv 1$; (b) $\ell_i$ follows an exponential distribution with mean 1; and (c) $\ell_i$ has mean one and it follows a mixture of an exponential distribution with mean 5 and an exponential distribution with mean 1/5. Three exciting functions $h(t)$ considered are: (a) $h_1(t)=\frac{9}{10}(1+t)^{-2}$; (b) $h_2(t)=\frac{9}{10}(1+t)^{-3}$; and (c) $h_3(t)=\frac{9}{10}e^{-t}$.}
\end{table} \label{table:time}

We now report numerical results for the performance quantities (b)--(e) in Section~\ref{sec:performance}, since the time-to-first-fill is completely determined by the baseline intensity function. Unless otherwise stated, we fix a constant baseline intensity $\mu(t)\equiv \mu=1$.

\textbf{Varying the trade size distribution $\ell_i$ while fixing an exciting function $h(t)= \frac{9}{10} \frac{1}{(1+t)^2}$.} Without loss of generality, we consider here three different distributions for $\ell_i$, all with a unit mean:
(a) $\ell_i \equiv 1$; (b) $\ell_i$ has an exponential distribution; and (c) $\ell_i$ has a hyper-exponential distribution: here, we consider a concrete example where
$\ell_i$ follows a mixture of an exponential distribution with mean 5 and an exponential distribution with mean 1/5. This choice of trade size distributions is motivated by \citet{Afeche2014}. In particular,
a mixture of exponential distributions with different means can capture the feature that in dark pools, impatient high frequency traders submit small ``pinging'' orders and liquidity traders may submit relatively larger orders. In addition, the class of hyper-exponential distributions is very rich that it can approximate many heavy-tailed distributions for trade sizes, while maintaining analytical tractability, see e.g. \citet{Cai2009, CaiKou2011}.

Figure~\ref{fig2} summarizes the results on the expected time-to-complete-fill as a function of the resting order size $x$.
Two key observations stand out from our results in Figure \ref{fig2}.
\begin{figure}[h]
	\centering		
	\includegraphics[width=0.6\textwidth]{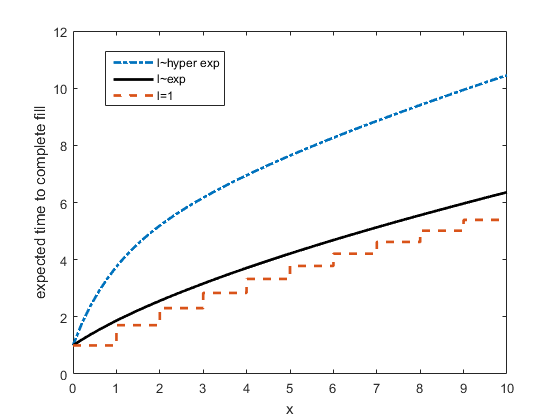}
	\caption{\emph{Expected time-to-complete-fill $\mathbb{E}[\sigma_x]$ in \eqref{eq:E-TTF}, as a function of the resting order size $x$. Here, $\mu(t) \equiv 1$ and the exciting function $h(t)= \frac{9}{10} \frac{1}{(1+t)^2}$ are fixed. The three curves correspond to three different distributions with a unit mean: (a) $\ell_i \equiv 1$ (red); (b) $\ell_i$ follows an exponential distribution (black); and (c) $\ell_i$ follows a mixture of an exponential distribution with mean 5 and an exponential distribution with mean 1/5 (blue). }}
	\label{fig2}
\end{figure}

First, the expected time-to-complete-fill of the dark order increases in the size $x$ of the dark order and changes significantly when the distribution of the incoming trade size varies. It can be seen from the figure that given the size of the investor's resting order $x$, when the incoming trade size follows a hyper-exponential distribution (a mixture of exponential distributions), it takes longer on average to completely fill this resting order than the cases of an exponential and a constant order size with the same mean. This observation is similar to the special case when $h \equiv 0$ where the point process $L$ becomes a compound Poisson process, and one can show that (see the Appendix for a proof)
\begin{equation}\label{321}
\mathbb{E}[\sigma^{(3)}_x]  > \mathbb{E}[\sigma^{(2)}_x] = x +1 > \mathbb{E}[\sigma^{(1)}_x] = \lceil x \rceil, \quad \text{for $x > 0.$}
\end{equation}
Here $\mathbb{E}[\sigma^{(i)}_x]$, $i=1,2,3$, are the expected time-to-complete-fill for
the compound Poisson arrival with trade sizes $\ell_{i}^{(1)}$ being constant, $\ell_{i}^{(2)}$ being exponential and $\ell_{i}^{(3)}$ being hyper-exponential (all with mean one) respectively.
We also remark that the expected time-to-complete-fill with Hawkes trades arrivals
depends on the distribution of $\ell_i$, not just its coefficient of variation.

Second, the expected time-to-complete-fill of the first unit of a resting dark order is greater than the second and subsequent units. This reflects the self-exciting modeling of the order execution process which captures the trade clustering behavior. In other words, after a partial execution of the resting dark order, the expectation of another trade and the future trading intensity will increase, which leads to a continuing reduction of the marginal time-to-complete-fill of the resting dark order.

Next in Figure~\ref{fig3}, we plot the expected fill rate of a resting order of size $x=10$, as a function of rest time $t$ for different trade size distributions. We observe that for a
hyper-exponential trade size distribution, the expected fill rate of the resting order is much smaller than the case of a constant order size with the same mean. This is
consistent with the observations from Figure~\ref{fig2} which suggest that it is harder to fill an order with ``more variable" trade sizes.  We provide an informal explanation on this relative order of expected fill rate for different trade size distributions in Figure~\ref{fig3}.
Let $L_{t}^{(1)}$, $L_{t}^{(2)}$, and $L_{t}^{(3)}$ denote
the associated $L_{t}$ when $\ell_{1}^{(1)}$ is a constant $1$, $\ell_{1}^{(2)}$ is exponentially distributed
and $\ell_{1}^{(3)}$ is hyper-exponentially distributed, all with the same mean.
First, notice that $\mathbb{E}[L_{t}^{(1)}]=
\mathbb{E}[L_{t}^{(2)}]=\mathbb{E}[L_{t}^{(3)}]$ from Proposition~\ref{prop:moments}
as $\mathbb{E}[\ell_{1}^{(1)}]=\mathbb{E}[\ell_{1}^{(2)}]=\mathbb{E}[\ell_{1}^{(3)}]$.
Next, observe that the hyper-exponential distribution is more ``spread out" than the exponential distribution which is more ``spread out" than a constant. Now when the jump size increases, the intensity of future arrivals
also increase, and as a result $L_{t}$ increases. Similar
argument holds when the jump size decreases.
This suggests $L^{(3)}_{t}$ with a hyper-exponential jump size is ``more variable" than $L^{(2)}_{t}$ with an exponential jump size in the sense that $L^{(3)}_{t}$ is more likely to take on ``extreme" values.
So intuitively, the expected fill rates
satisfy $\frac{1}{x}\mathbb{E}[\min\{L_{t}^{(1)},x\}]
\geq \frac{1}{x}\mathbb{E}[\min\{L_{t}^{(2)},x\}]\geq \frac{1}{x}\mathbb{E}[\min\{L_{t}^{(3)},x\}]$.

\begin{figure}[h]
	\centering		
	\includegraphics[width=0.6\textwidth]{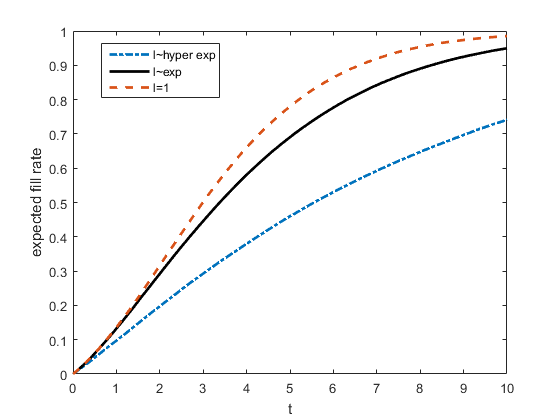}
	\caption{\emph{Expected fill rate $\mathbb{E} [\min\{L_{t},x\}]/x$ in \eqref{eq:fillrate} vs rest time $t$, for a resting order with size $x=10$. Here, $\mu(t)\equiv 1$ and the exciting function $h(t)= \frac{9}{10} \frac{1}{(1+t)^2}$ are fixed.  The three curves correspond to three different distributions with a unit mean: (a) $\ell_i \equiv 1$ (red); (b) $\ell_i$ follows an exponential distribution (black); and (c) $\ell_i$ follows a mixture of an exponential distribution with mean 5 and an exponential distribution with mean 1/5 (blue). }}
	\label{fig3}
\end{figure}

Furthermore, we plot in Figure~\ref{fig4} the conditional probability of one fill and the conditional expected fill size for the resting order as a function of the future $T$ units of time, given that there is a fill of size one in the past two units of time. Mathematically, the event conditioned on is $\{N_2=1,\ell_1=1\}$.
\begin{figure}[h]
	\centering
	\subfigure[conditional probability of one fill]{
		\begin{minipage}[b]{0.45\textwidth}
			\centering
			\includegraphics[width=\textwidth]{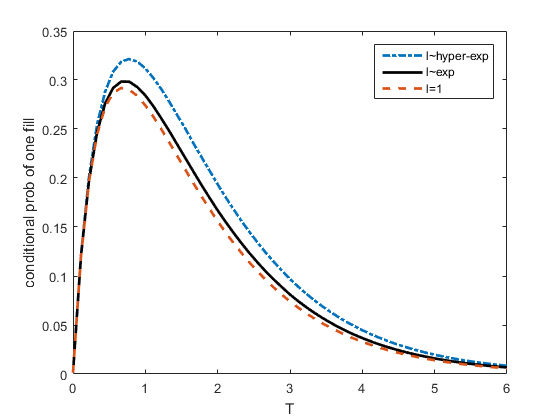}
	\end{minipage}	}			
	\subfigure[conditional expected fill size]{
		\begin{minipage}[b]{0.45\textwidth}
			\centering
			\includegraphics[width=\textwidth]{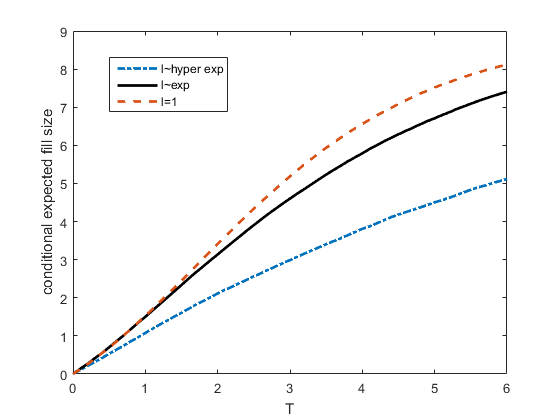}
	\end{minipage}}	
	\caption{\emph{(a) The probability of one fill in \eqref{eq:cond-hit} and (b) the expected fill size in \eqref{eq:conditional-expec} for a resting order with size $x=10$, conditioning on the event $\{N_2=1,\ell_1=1\}$. Here, $\mu(t) \equiv 1$ and the exciting function $h(t)= \frac{9}{10} \frac{1}{(1+t)^2}$ are fixed. The three curves correspond to three different distributions with a unit mean: (a) $\ell_i \equiv 1$ (red); (b) $\ell_i$ follows an exponential distribution (black); and (c) $\ell_i$ follows a mixture of an exponential distribution with mean 5 and an exponential distribution with mean 1/5 (blue).}}
	\label{fig4}
\end{figure}

For the conditional probability of one fill in Figure~\ref{fig4}(a), we note that
it is biggest when the trade size $\ell_i$ follows a hyper-exponential distribution with mean one, and it is smallest when the trade size is constantly one. Let us explain. It is clear from \eqref{eq:cond-prob2} that this conditional probability of one fill depends monotonically on the following Laplace transform of the random variable $\ell_1$:
\begin{align}\label{eq:lap-h}
\mathbb{E} \left[e^{-\int_{0}^{T+t-s} h(u)du \cdot \ell_1}\right].
\end{align}
If we denote $\alpha:=\int_{0}^{T+t-s} h(u)du \ge 0$, then computing the Laplace transform in \eqref{eq:lap-h} for the three distributions of $\ell_1$ (hyper-exponential, exponential and constant one) yields
\[ \frac{1}{6} \cdot \frac{0.2}{0.2 +\alpha } + \frac{5}{6} \cdot \frac{5}{5 +\alpha } \ge \frac{1}{1+\alpha} \ge e^{-\alpha}. \]
Now the observation in Figure~\ref{fig4}(a) follows from \eqref{eq:cond-prob2} and the above inequalities.  

\textbf{Varying the exciting function while fixing an exponential trade size distribution.} We next investigate how the exciting function $h$ impacts the performance quantities. For illustration purposes, we consider a family of power-law exciting functions with different tail behaviors:
\begin{equation} \label{eq:h-gamma}
h^{\gamma}(t) = \frac{C}{(1+t)^{\gamma}}, \quad \text{for $C>0, \quad \gamma>1$}.
\end{equation}
In particular, $\Vert h^{\gamma}\Vert_{L^{1}}=\int_{0}^{\infty}h^{\gamma}(t)dt = \frac{C}{\gamma -1}$. In the literature, this quantity $\Vert h^{\gamma}\Vert_{L^{1}}$ is usually interpreted as a branching ratio, i.e., the expected number of events generated by any parent event. In the following, we will fix $C=0.9$, and $\ell_i$ follows an exponential distribution with mean 1 and the baseline intensity $\mu=1$. We consider three different exciting functions $h^{\gamma}(t)= \frac{9}{10} \frac{1}{(1+t)^\gamma}$ corresponding to $\gamma=2,2.5$ and 3. This allows us to better understand how the exciting functions impact the performance quantities.

We first plot in Figure~\ref{fig5} the expected time-to-complete-fill of an order as a function of the order size $x$, for different exciting functions $h^{\gamma}$ in \eqref{eq:h-gamma}. As one can observe from Figure~\ref{fig5}, the larger the $\gamma$,
the longer it takes on average to fill a given order completely.

\begin{figure}[h]
	\centering		
	\includegraphics[width=0.6\textwidth]{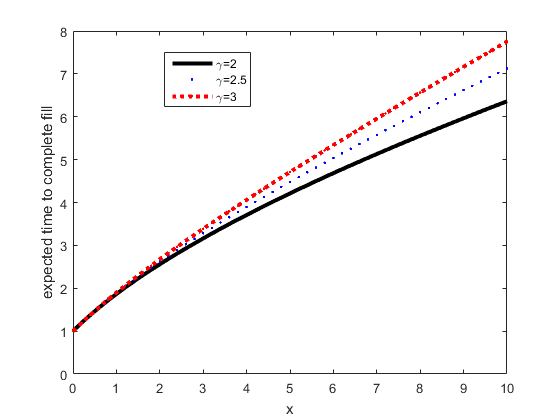}
	\caption{\emph{Expected time-to-complete-fill $\mathbb{E}[\sigma_x]$ in \eqref{eq:E-TTF} vs order size $x$ for different exciting functions $h^{\gamma}$ defined in \eqref{eq:h-gamma}. Here, $\ell_i$ follows an exponential distribution with mean 1 and the baseline intensity $\mu(t) \equiv 1$ for the Hawkes model.}}
	\label{fig5}
\end{figure}

Next, we plot in Figure~\ref{fig6} the expected fill rate for a given order with size $x=10$, as a function of the resting time of the order. Consistent with Figure~\ref{fig5}, the larger the $\gamma,$ the harder to fill an order and hence the smaller the expected fill rate for a given resting time.

\begin{figure}[h]
	\centering		
	\includegraphics[width=0.6\textwidth]{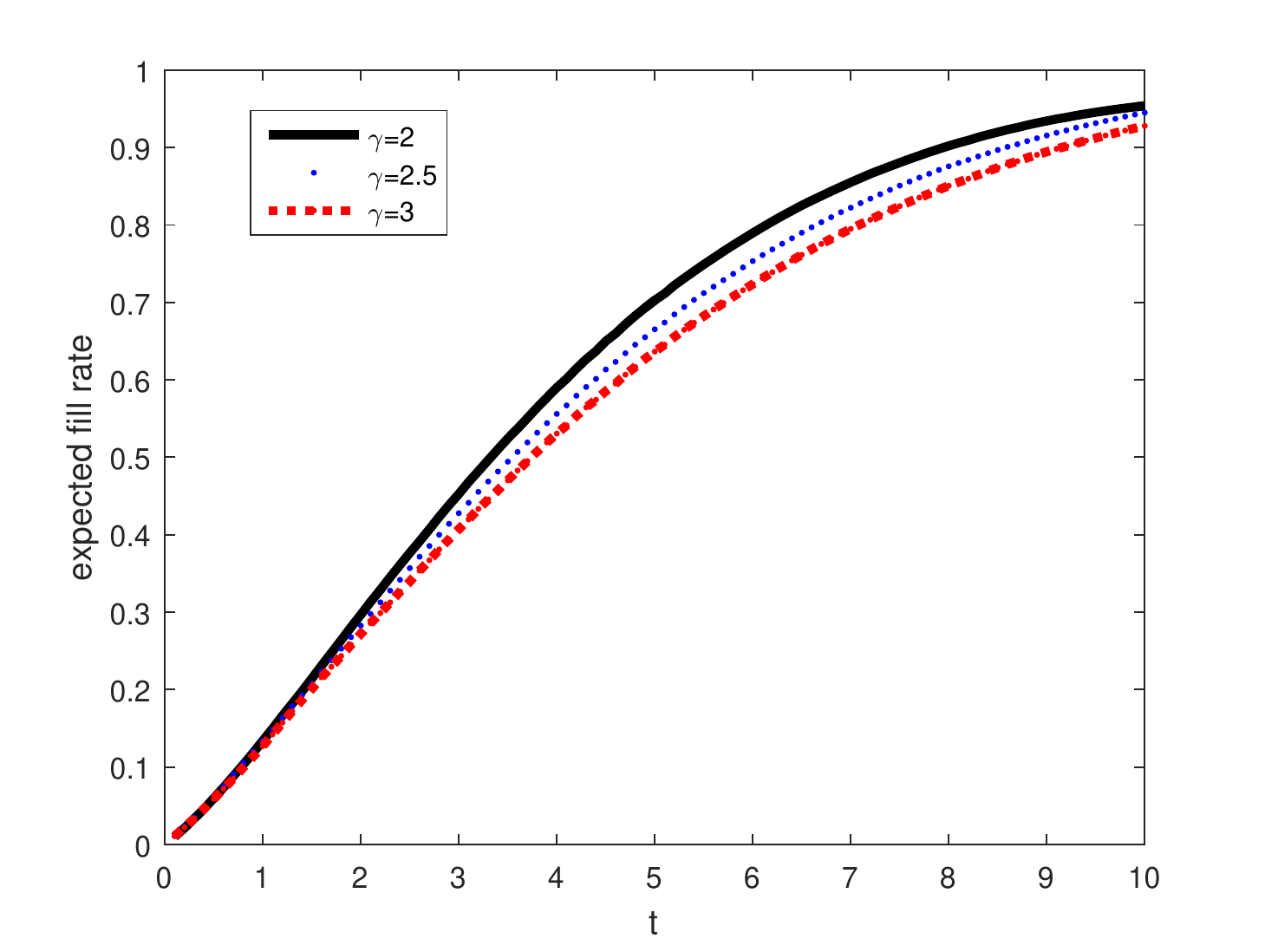}
	\caption{\emph{Expected fill rate $\mathbb{E} [\min\{L_{t},x\}]/x$ in \eqref{eq:fillrate} vs rest time $t$ for a given order with size $x=10$. Here, $\ell_i$ follows an exponential distribution with mean 1 and the baseline intensity $\mu=1$ for the Hawkes model. The three curves correspond to three different exciting functions $h^{\gamma}$ defined in \eqref{eq:h-gamma}.}}
	\label{fig6}
\end{figure}

Let us explain the phenomenon observed in Figures~\ref{fig5} and \ref{fig6}. For $\gamma_1>\gamma_2>1$, we find from \eqref{eq:h-gamma} that $h^{\gamma_1}(t)<h^{\gamma_2}(t)$ for all $t\ge 0$. This implies that one can find a common probability space such that the associated point processes satisfy $L_t^{\gamma_1}\leq L_t^{\gamma_2}$ for all $t$ almost surely. Then the observations in Figures~\ref{fig5} and \ref{fig6} readily follow from the formulas for the expected time-to-complete-fill in \eqref{eq:E-TTF} and the expected fill rate in \eqref{eq:fillrate}.

We further investigate the conditional probability of another fill and the conditional expected fill size for different exciting functions, for a given resting order of size 10. Again, the event conditioned on is $\{N_2=1,\ell_1=1\}$, i.e., there is a fill of size one in the past two units of time. These two performance quantities are plotted in Figure~\ref{fig7}. We find from Figure~\ref{fig7}(a) that, unlike in Figure~\ref{fig4}(a), there is no monotonicity for the conditional probability of one fill when we vary $\gamma.$ This is not surprising as we can see from the formula \eqref{eq:cond-prob2} that this conditional probability depends on the exciting function in a delicate way.

\begin{figure}[h]
	\centering
	\subfigure[conditional probability of one fill]{
		\begin{minipage}[b]{0.45\textwidth}
			\centering
			\includegraphics[width=\textwidth]{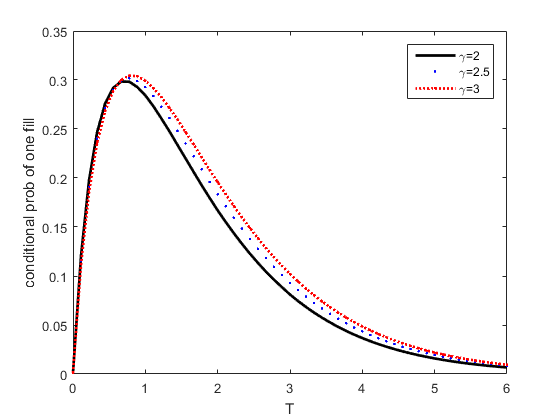}
	\end{minipage}	}			
	\subfigure[conditional expected fill size]{
		\begin{minipage}[b]{0.45\textwidth}
			\centering
			\includegraphics[width=\textwidth]{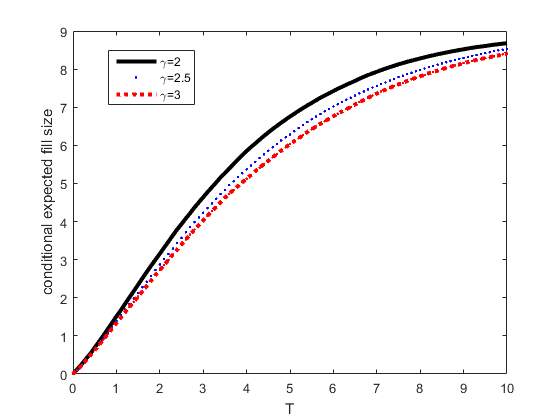}
	\end{minipage}}	
	\caption{\emph{(a) The probability of one fill in \eqref{eq:cond-hit} and (b) the expected fill size \eqref{eq:conditional-expec} for a resting order of size $x=10$, conditioning on the event $\{N_2=1,\ell_1=1\}$, i.e., there is a fill of size one in the past two units of time. Here, $\ell_i$ follows an exponential distribution with mean 1 and the baseline intensity $\mu=1$ for the Hawkes model. The three curves correspond to three different exciting functions $h^{\gamma}$ defined in \eqref{eq:h-gamma}.}}
	\label{fig7}
\end{figure}

\textbf{The effect of the baseline intensity $\mu(t)$ on performance metrics.}
So far, all the numerical examples on performance metrics are presented assuming a constant baseline intensity $\mu(t) \equiv 1$. We now
briefly discuss the effect of a non-constant baseline intensity of the Hawkes process on performance metrics. Such a time-dependent baseline intensity $\mu(t)$ could represent, for example, the intraday pattern of dark pool liquidity.

For illustration purposes, we focus on the representative performance metric: the expected fill rate of a resting dark order given in \eqref{eq:fillrate}. We compare in Figure~\ref{fig-mut} the case of a constant baseline intensity $\mu(t) \equiv 1$ with the following two cases where $\mu(t)$ is piecewise constant:
\begin{align}\label{eq:mu-t}
  \mu_1(t) =
  \begin{cases}
    2, & \text{for } 0 \leq t \leq 4, \\
   0.5 , & \text{for } 4 < t \leq 8, \\
   1, & \text{for } t >8.
  \end{cases}  \quad \text{and} \quad
  \mu_2(t) =
  \begin{cases}
    0.5, & \text{for } 0 \leq t \leq 4, \\
   2 , & \text{for } 4 < t \leq 8, \\
   1, & \text{for } t > 8.
  \end{cases}
\end{align}
We can observe from Figure~\ref{fig-mut} that the expected fill rate of a resting dark order depends on the baseline intensity of the Hawkes execution process in a delicate way. For the initial time period $[0, 4]$, as $\mu_1(t)>1> \mu_2(t)$, it follows that a higher baseline intensity of the Hawkes execution process leads to a higher expected fill rate of the resting order. On the other hand, on the time interval $(4,8]$, we have $\mu_2(t)>1>\mu_1(t)$.
Compared with a Hawkes execution process with a constant one baseline intensity, the expected fill rate of a resting dark order during the time interval $(4, 8]$ is still higher when the trades follow a Hawkes arrival process with a baseline intensity $\mu_1(t)<1$. In addition, the expected fill rate of an order may also become higher when the Hawkes process has a baseline intensity $\mu_2(t)>1$. These two observations are due to the fact the intensity of a Hawkes process depends on both the baseline intensity and its own entire history (i.e. the past occurrence of trades).

\begin{figure}[h]
	\centering		
	\includegraphics[width=0.6\textwidth]{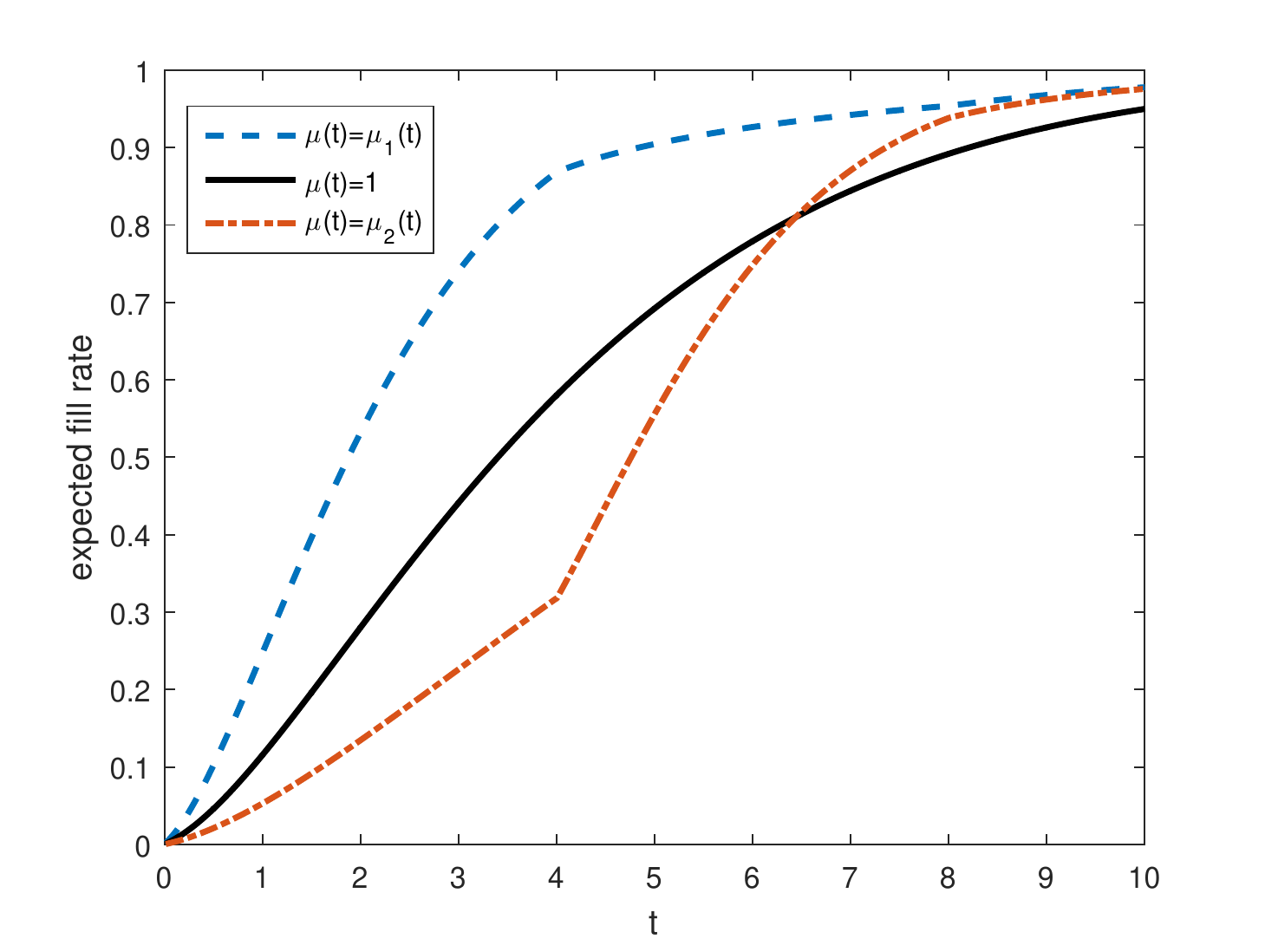}
	\caption{\emph{Expected fill rate $\mathbb{E} [\min\{L_{t},x\}]/x$ in \eqref{eq:fillrate} vs rest time $t$ for a given order with size $x=10$. Here, $\ell_i$ follows an exponential distribution with mean 1 and the exciting function $h(t)=\frac{9}{10} \frac{1}{(1+t)^2}$ for the Hawkes model. The three curves correspond to three different baseline intensity functions $\mu(t)$: $\mu(t) \equiv 1$ and $\mu_1(t), \mu_2(t)$ given in \eqref{eq:mu-t}.}}
	\label{fig-mut}
\end{figure}

\subsection{Non-empty dark pools} \label{sec:nonempty}
The performance formulas derived in Section~\ref{sec:performance} can be generalized to non-empty dark pools. To illustrate, we consider computing time-to-first-fill,
time-to-complete-fill and expected fill rate of a posted dark order.

Suppose at time zero when the investor's midpoint peg buy order of size $x>0$ reaches the dark pool, the liquidity size $Y$ in the pool is a random variable with a known or estimated cumulative distribution function $F_Y(y)$ which could possibly have a mass at zero (see e.g. \citet{Ganchev2010}). Here $Y>0$ represents the size of existing buy orders at the midpoint, and $Y<0$ represents the size of sell orders resting at the midpoint in the pool. In particular, there will be an immediate execution of the investor's buy order at time zero if $Y<0$. In this case, when $Y\le -x$, then the investor's dark order is completely filled at time zero. Otherwise for $Y \in (-x, 0)$, the dark buy order will get partially filled. The trader on the other side
of the completed trade then realizes there could potentially be
more liquidity on the opposite side of his trade, and then re-routes his other orders to this pool. Other information-seekers may also notice the trade and submit orders to this pool. Hence, this trade against resting sell orders at time zero may also incur a jump of the intensity of the arriving sell trades.

 Mathematically, with a random liquidity size $Y$ in the dark venue, the intensity of the Hawkes process $N$ modeling the executions of the dark buy order will be modified as follows (defined till the time the dark order is completely filled):
\begin{equation} \label{eq:modifydynamics}
\lambda_{t}=\mu+ \min\{x, |Y|\}\cdot 1_{Y<0} \cdot h(t) + \sum_{0<\tau_i<t}  h(t-\tau_i) \cdot \ell_i ,
\end{equation}
where $\min\{x, |Y|\}\cdot 1_{Y<0}$ represents the size of a fill at time zero, and the impact on the trading intensity also decays according to the exciting function $h$.
Hence, conditioned on $Y=y>0$, suppose these existing buy orders with total size $y$ also rest in the pool until full execution, then the Hawkes process $(N,L)$ will be essentially the same as in the case of an empty dark pool. On the other hand, conditioned on $Y=y<0$, the intensity $\lambda$ follows a different dynamics where now the baseline intensity is time-varying as given in the first two parts of the expression \eqref{eq:modifydynamics}.

We now derive the performance quantities and still use the same notations as in Section~\ref{sec:performance} for simplicity. First, with a random liquidity size $Y$ in the dark pool,
the time-to-first-fill of the investor's dark buy order is given by
\begin{equation*}
\tau(1)=\inf\{t \ge 0:L_t >Y\}.
\end{equation*}
Hence we obtain
\begin{eqnarray*}
\mathbb{P}(\tau(1)=0) &=& \mathbb{P}(Y<0)= F_Y(0-),\\
\mathbb{P}(\tau(1)>t)& =& \mathbb{P}(L_t \le Y) = \int_0^{\infty}\mathbb{P}(L_t \le y) dF_{Y}(y), \quad \text{for $t \ge 0.$}
\end{eqnarray*}
Since we have derived the transform of $L_t$, we can then use inverse Laplace transform to get $\mathbb{P}(L_t \le y)$ for $y>0$, and hence compute the distribution and the expectation of $\tau(1)$.

Second, the time-to-complete-fill of the investor's dark buy order with size $x>0$ is given by
 \[\sigma_x=\inf\{t \ge 0: L_t\ge x +Y \},\]
hence we have
\begin{eqnarray}
\mathbb{P}(\sigma_x=0) &=& \mathbb{P}(Y \le -x) =F_Y(-x), \label{eq:tau-x-1}\\
\mathbb{P}(\sigma_x>t) & =& \mathbb{P}(L_t < x+ Y) , \quad \text{for $t \ge 0.$}  \label{eq:tau-x-2}
\end{eqnarray}
To obtain the distribution of $\sigma_x$, it suffices to compute $\mathbb{P}(L_t < x+ Y)$. Note that
\begin{eqnarray}
\mathbb{P}(L_t < x+ Y) & =& \mathbb{P}(L_t -Y < x) \nonumber \\ &=& \int_{-x}^{0}\mathbb{P}(L_t < x + y| Y=y) dF_{Y}(y)  \nonumber \\
&&\qquad\qquad\qquad+\int_{0}^{\infty}\mathbb{P}(L_t < x + y) dF_{Y}(y). \label{eq:Lt-Y}
\end{eqnarray}
So it only remains to compute $\mathbb{P}(L_t < x + y| Y=y)$ for $-x < y <0.$ Conditioned on $Y=y \in (-x, 0)$, i.e., there is a partial execution of the dark buy order at time zero (which matches with resting sell orders in the pool), then from \eqref{eq:modifydynamics} we infer that the intensity of the Hawkes process becomes
\begin{equation*}
\lambda_{t}=\mu + |y| \cdot h(t) + \sum_{0<\tau_i<t}  h(t-\tau_i) \cdot \ell_i ,
\end{equation*}
where the baseline intensity $\mu(t)=\mu + |y| \cdot h(t)$ is deterministic and time-dependent.
Since we have computed the Laplace transform of $(N_t, L_t)$ where the baseline intensity of the Hawkes process can be time-varying, we can then use inverse Laplace transform to compute $\mathbb{P}(L_t < x + y| Y=y)$ for $-x < y <0$. Now
the distribution of time-to-complete-fill of the midpoint dark order can be computed using \eqref{eq:tau-x-1}, \eqref{eq:tau-x-2} and \eqref{eq:Lt-Y}. The expected time-to-complete-fill $\mathbb{E}[\sigma_x]$ also readily follows.

Finally, the expected fill rate of the investor's resting midpoint dark order of size $x$, in a given time interval $[0, t]$, is given by
  \begin{equation}\label{eq:fillrate-nonempty}
 \frac{1}{x} \cdot \mathbb{E} [\min\{(L_{t}-Y)^+,x\}] = \frac{1}{x} \cdot {\int_{0}^x \mathbb{P}((L_{t}-Y)^+ >z)dz},
 \end{equation}
where $a^+:=\max\{a,0\}$ for a real number $a$. This expected fill rate can hence also be readily computed as we have derived the distribution of $L_{t}-Y$ in \eqref{eq:Lt-Y}.

 Therefore, if the dark pool is non-empty at the time of the dark order placement by the investor, these performance quantities can still be similarly derived and efficiently numerically computed.

 \subsubsection*{Numerical examples}
 For illustration purposes, we plot in Figure~\ref{fig8} the expected fill rate of a given midpoint peg order when the initial liquidity size $Y$ in the dark pool has the following distribution: $\mathbb{P}(Y=0)=0.3$ and when $Y \ne 0$, it has a density function
 \begin{equation*}
 f_Y(y)=0.35\cdot k |y|^{k-1} e^{-|y|^k}, \quad \text{for $y \ne 0.$}
 \end{equation*}
 That is, $Y$ has a mass at zero, and it follows a two-sided Weibull distribution with scale parameter 1 and shape parameter $k$. Note that when $k \in (0,1),$ the tail of $Y$ is heavier than that of a two-sided exponential distribution. In addition, a smaller shape parameter $k$ implies a heavier tail of the liquidity size $Y$.
Such a choice of $Y$ is motivated by \citet{Ganchev2010} which empirically finds that the distribution of volume in dark pools is heavy-tailed: often, there is no volume available, but sometimes very large volume is present.

\begin{figure}[h]
	\centering		
	\includegraphics[width=0.6\textwidth]{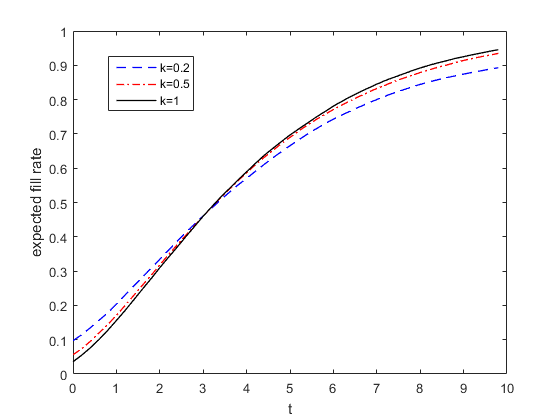}
	\caption{\emph{Expected fill rate in \eqref{eq:fillrate-nonempty} vs rest time $t$ for a given resting order with size $x=10$. Here, $\mu=1$, the trade size $\ell_i$ follows an exponential distribution with mean 1 and the exciting function $h(t)= \frac{9}{10} \frac{1}{(1+t)^2}$. The initial liquidity in the pool is modeled by a random variable $Y$ which has a mass 0.3 at zero, and when $Y \ne 0$, it follows a two-sided Weibull distribution with scale parameter 1 and shape parameter $k.$ }}
	\label{fig8}
\end{figure}

We observe from Figure~\ref{fig8} that the expected fill rate of a posted order at $t=0$ in a non-empty pool is greater than zero, which is different from that in an empty pool. Intuitively, this is clear as there could be contra-side sell orders resting at the midpoint in a non-empty pool which triggers trades at time zero when the investor posts a buy order at the midpoint. In fact, mathematically we can deduce from \eqref{eq:fillrate-nonempty} that the expected fill rate at time zero is simply
  \begin{equation*} \label{eq:fill-0}
\frac{1}{x} \cdot {\int_{0}^x \mathbb{P}((-Y)^+ >z)dz},
 \end{equation*}
since $L_0=0.$ Hence, the initial percentage of fill when the order is posted critically depends on both the posted order size $x$, and the distribution of the volume of contra-side resting orders $(-Y)^+$. While we also observe from Figure~\ref{fig8} that a smaller shape parameter $k$ of the liquidity $Y$ leads to a larger expected fill rate at time zero, we remark that our extensive numerical experiments show this is not generally true for any order size $x$.

\section{Conclusions}\label{sec:conclusion}

We study the Hawkes process, a self-exciting point process, where
the baseline intensity is time-dependent, the exciting function is a general function and the jump sizes of the intensity process are i.i.d. non-negative random variables.
We obtain closed-form formulas for the Laplace transform, moments and the distribution of the Hawkes process. We apply these results to dark pool trading and analyze various performance metrics of a resting dark order which trades against contra-side marketable
orders arriving according to a Hawkes process. These performance quantities can be useful for strategic allocation of liquidity among different pools to reduce market impact and execution costs in portfolio trading.

%
%



\appendix
\section{Proof of Proposition~\ref{prop:moments}}
\begin{proof}[Proof of Proposition~\ref{prop:moments}] We first compute the first two moments of the counting process $N$.
By differentiating the Laplace transform of the counting process $N$
with respect to (w.r.t.) $\theta$ in Corollary \ref{SingleTransform}, we get
\begin{equation*}
\frac{\partial}{\partial\theta}\mathbb{E}[e^{-\theta N_{T}}]
=\int_{0}^{T}\mu(T-s)\frac{\partial}{\partial\theta}F_{N}(s)ds
e^{\int_{0}^{T}\mu(T-s)(F_{N}(s)-1)ds},
\end{equation*}
and by differentiating w.r.t. $\theta$ again, we get
\begin{align*}
\frac{\partial^{2}}{\partial\theta^{2}}\mathbb{E}[e^{-\theta N_{T}}]
&=\int_{0}^{T}\mu(T-s)\frac{\partial^{2}}{\partial\theta^{2}}F_{N}(s)ds
e^{\int_{0}^{T}\mu(T-s)(F_{N}(s)-1)ds}
\\
&\qquad\qquad
+\left(\int_{0}^{T}\mu(T-s)\frac{\partial}{\partial\theta}F_{N}(s)ds\right)^{2}
e^{\int_{0}^{T}\mu(T-s)(F_{N}(s)-1)ds}.
\nonumber
\end{align*}
By differentiating both sides of \eqref{FN} w.r.t. $\theta$, we get
\begin{equation}\label{FNDiff1}
\frac{\partial}{\partial\theta}F_{N}(t)
=\mathbb{E}\left[\left(-1+\int_{0}^{t}\ell_{1}h(s)\frac{\partial}{\partial\theta}
F_{N}(t-s)ds\right)e^{-\theta+\int_{0}^{t}\ell_{1}h(s)(F_{N}(t-s)-1)ds}\right].
\end{equation}
By differentiating w.r.t. $\theta$ again, we get
\begin{align}\label{FNDiff2}
\frac{\partial^{2}}{\partial\theta^{2}}F_{N}(t)
&=\mathbb{E}\left[\left(-1+\int_{0}^{t}\ell_{1}h(s)\frac{\partial}{\partial\theta}
F_{N}(t-s)ds\right)^{2}e^{-\theta+\int_{0}^{t}\ell_{1}h(s)(F_{N}(t-s)-1)ds}\right]
\\
&\qquad\qquad
+\mathbb{E}\left[\int_{0}^{t}\ell_{1}h(s)\frac{\partial^{2}}{\partial\theta^{2}}
F_{N}(t-s)ds
e^{-\theta+\int_{0}^{t}\ell_{1}h(s)(F_{N}(t-s)-1)ds}\right].
\nonumber
\end{align}
By letting $\theta=0$ in \eqref{FNDiff1}, we get
\begin{equation*}
\frac{\partial}{\partial\theta}F_{N}(t)\bigg|_{\theta=0}
=-1+\int_{0}^{t}\mathbb{E}[\ell_{1}]h(s)\frac{\partial}{\partial\theta}
F_{N}(t-s)\bigg|_{\theta=0}ds.
\end{equation*}
This implies that
\begin{equation*}
\frac{\partial}{\partial\theta}F_{N}(t)\bigg|_{\theta=0}=-\psi_{1}(t),
\end{equation*}
where $\psi_{1}(\cdot)$ is defined in \eqref{psi1} and thus
\begin{equation*}
\mathbb{E}[N_{T}]=-\frac{\partial}{\partial\theta}\mathbb{E}[e^{-\theta N_{T}}]\bigg|_{\theta=0}
=-\int_{0}^{T}\mu(T-s)\frac{\partial}{\partial\theta}F_{N}(s)\bigg|_{\theta=0}ds
=\int_{0}^{T}\mu(T-s)\psi_{1}(s)ds.
\end{equation*}
By letting $\theta=0$ in \eqref{FNDiff2}, we get
\begin{equation*}
\frac{\partial^{2}}{\partial\theta^{2}}F_{N}(t)\bigg|_{\theta=0}
=(\psi_{1}(t))^{2}
+\int_{0}^{t}\mathbb{E}[\ell_{1}]h(s)\frac{\partial^{2}}{\partial\theta^{2}}F_{N}(t-1)\bigg|_{\theta=0}ds.
\end{equation*}
By the definition of $\psi_{2}(\cdot)$ in \eqref{psi2}, we have
$\frac{\partial^{2}}{\partial\theta^{2}}F_{N}(t)\big|_{\theta=0}=\psi_{2}(t)$.
Finally, we conclude that
\begin{align*}
\mathbb{E}[N_{T}^{2}]
&=\frac{\partial^{2}}{\partial\theta^{2}}\mathbb{E}[e^{-\theta N_{T}}]\bigg|_{\theta=0}
\\
&=\int_{0}^{T}\mu(T-t)\frac{\partial^{2}}{\partial\theta^{2}}F_{N}(s)\bigg|_{\theta=0}ds
+\left(\int_{0}^{T}\mu(T-s)\frac{\partial}{\partial\theta}F_{N}(s)\bigg|_{\theta=0}ds\right)^{2}
\nonumber
\\
&=\int_{0}^{T}\mu(T-t)\psi_{2}(t)dt+\left(\int_{0}^{T}\mu(T-s)\psi_{1}(t)dt\right)^{2}.
\nonumber
\end{align*}

We next compute the first two moments of the process $L$. We can compute from \eqref{FL} that
\begin{equation*}
\frac{\partial}{\partial\theta}
F_{L}(t)\bigg|_{\theta=0}
=-\mathbb{E}[\ell_{1}]+\int_{0}^{t}\mathbb{E}[\ell_{1}]h(s)\frac{\partial}{\partial\theta}
F_{L}(t-s)\bigg|_{\theta=0}ds,
\end{equation*}
which implies that
\begin{equation*}
\frac{\partial}{\partial\theta}
F_{L}(t)\bigg|_{\theta=0}
=\mathbb{E}[\ell_{1}]\psi_{1}(t).
\end{equation*}
Hence,
\begin{equation*}
\mathbb{E}[L_{T}]=-\frac{\partial}{\partial\theta}\mathbb{E}[e^{-\theta L_{T}}]\bigg|_{\theta=0}
=-\int_{0}^{T}\mu(T-s)\frac{\partial}{\partial\theta}F_{L}(s)\bigg|_{\theta=0}ds
=\mathbb{E}[\ell_{1}]\int_{0}^{T}\mu(T-s)\psi_{1}(s)ds.
\end{equation*}
We can also compute from \eqref{FL} that
\begin{align*}
\frac{\partial^{2}}{\partial\theta^{2}}F_{L}(t)
&=\mathbb{E}\left[\left(-\ell_{1}+\int_{0}^{t}\ell_{1}h(s)\frac{\partial}{\partial\theta}
F_{L}(t-s)ds\right)^{2}e^{-\theta\ell_{1}+\int_{0}^{t}\ell_{1}h(s)(F_{L}(t-s)-1)ds}\right]
\\
&\qquad\qquad
+\mathbb{E}\left[\int_{0}^{t}\ell_{1}h(s)\frac{\partial^{2}}{\partial\theta^{2}}
F_{L}(t-s)ds
e^{-\theta\ell_{1}+\int_{0}^{t}\ell_{1}h(s)(F_{L}(t-s)-1)ds}\right].
\nonumber
\end{align*}
Therefore, by the definition of $\psi_{1}$,
\begin{align*}
\frac{\partial^{2}}{\partial\theta^{2}}F_{L}(t)\bigg|_{\theta=0}
&=\mathbb{E}[\ell_{1}^{2}]
\left(-1+\int_{0}^{t}h(s)\frac{\partial}{\partial\theta}
F_{L}(t-s)\bigg|_{\theta=0}ds\right)^{2}
+\int_{0}^{t}\mathbb{E}[\ell_{1}]h(s)\frac{\partial^{2}}{\partial\theta^{2}}
F_{L}(t-s)\bigg|_{\theta=0}ds
\\
&=\mathbb{E}[\ell_{1}^{2}](\psi_{1}(t))^{2}
+\int_{0}^{t}\mathbb{E}[\ell_{1}]h(s)\frac{\partial^{2}}{\partial\theta^{2}}
F_{L}(t-s)\bigg|_{\theta=0}ds.
\end{align*}
Now, by recalling that $\psi_{3}(t)= \frac{\partial^{2}}{\partial\theta^{2}}F_{L}(t)\bigg|_{\theta=0}$, we conclude that
\begin{align*}
\mathbb{E}[L_{T}^{2}]
&=\frac{\partial^{2}}{\partial\theta^{2}}\mathbb{E}[e^{-\theta L_{T}}]\bigg|_{\theta=0}
\\
&=\int_{0}^{T}\mu(T-s)\frac{\partial^{2}}{\partial\theta^{2}}F_{L}(s)\bigg|_{\theta=0}ds
+\left(\int_{0}^{T}\mu(T-s)\frac{\partial}{\partial\theta}F_{L}(s)\bigg|_{\theta=0}ds\right)^{2}
\nonumber
\\
&=\int_{0}^{T}\mu(T-t)\psi_{3}(t)dt+(\mathbb{E}[\ell_{1}])^{2}\left(\int_{0}^{T}\mu(T-t)\psi_{1}(t)dt\right)^{2}.
\nonumber
\end{align*}

Finally, let us show that \eqref{psi1}, \eqref{psi2} and \eqref{psi3}
have unique solutions. We will only show uniqueness for the solution of \eqref{psi1} here,
and the argument is the same for \eqref{psi2} and \eqref{psi3}.
Assume that \eqref{psi1} has two solutions, say $\psi_{1}^{(1)}$ and $\psi_{1}^{(2)}$.
Then, for any $0\leq t\leq T$, we have
\begin{align*}
\left|\psi_{1}^{(1)}(t)-\psi_{1}^{(2)}(t)\right|
&\leq\int_{0}^{t}h(t-s)\mathbb{E}[\ell_{1}]\left|\psi_{1}^{(1)}(s)-\psi_{1}^{(2)}(s)\right|ds
\\
&\leq\Vert h\Vert_{L^{\infty}[0,T]}\mathbb{E}[\ell_{1}]\int_{0}^{t}\left|\psi_{1}^{(1)}(s)-\psi_{1}^{(2)}(s)\right|ds.
\end{align*}
By Gronwall's inequality, we conclude that $\psi_{1}^{(1)}=\psi_{1}^{(2)}$ on $[0,T]$.
\end{proof}

\section{Proof of Proposition~\ref{prop:mass}}

\begin{proof}[Proof of Proposition~\ref{prop:mass}]
Note that for any $z\in\mathbb{C}$ with $|z|\leq 1$,
by considering $z=e^{-\theta}$, we obtain from the Laplace transform of $N_T$ that
\begin{equation}\label{z1}
\mathbb{E}[z^{N_{T}}]=e^{\int_{0}^{T}\mu(T-s)(F_{N}(s)-1)ds},
\end{equation}
where (with slight abuse of notations) $F_N(t)$ depends on $z$ and it is given by
\begin{equation}\label{z2}
F_{N}(t)=z\cdot \mathbb{E}\left[e^{\int_{0}^{t}\ell_{1}h(s)(F_{N}(t-s)-1)ds}\right], \quad \text{for any $0\leq t\leq T$.}
\end{equation}
It is easy to see that
\begin{equation*}
\mathbb{E}[z^{N_{T}}]
=\sum_{k=0}^{\infty}z^{k}\mathbb{P}(N_{T}=k),
\end{equation*}
and hence
\begin{equation*}
\mathbb{P}(N_{T}=k)=\frac{1}{k!}\frac{\partial^{k}}{\partial z^{k}}\mathbb{E}[z^{N_{T}}]\bigg|_{z=0}.
\end{equation*}
Let us recall the celebrated Fa\`{a} di Bruno's formula, for any smooth functions $f$ and $g$:
\begin{equation}\label{BrunoFormula}
\frac{d^{n}}{dx^{n}}f(g(x))
=\sum\frac{n!}{m_{1}!1!^{m_{1}}m_{2}!2!^{m_{2}}\cdots m_{n}!n!^{m_{n}}}
f^{(m_{1}+\cdots+m_{n})}(g(x))\prod_{j=1}^{n}(g^{(j)}(x))^{m_{j}},
\end{equation}
where the summation is over all $n$-tuples of nonnegative integers $(m_1, \ldots, m_n)$ satisfying the constraint $1\cdot m_{1}+2\cdot m_{2}+3\cdot m_{3}+\cdots+n\cdot m_{n}=n$.
Notice that
\begin{equation*}
\mathbb{E}[z^{N_{T}}]=e^{-\int_{0}^{T}\mu(T-s)ds}e^{\int_{0}^{T}\mu(T-s)F_{N}(s)ds},
\end{equation*}
and $F_{N}=0$ for $z=0$. By applying Fa\`{a} di Bruno's formula \eqref{BrunoFormula}
(with $f(z)=e^{z}$, $g(z)=\int_{0}^{T}\mu(T-s)F_{N}(s)ds$ and $n=k$), we get:
\begin{align*}
\frac{\partial^{k}}{\partial z^{k}}\mathbb{E}[z^{N_{T}}]\bigg|_{z=0}
&=e^{-\int_{0}^{T}\mu(T-s)ds}\sum\frac{k!}{m_{1}!1!^{m_{1}}m_{2}!2!^{m_{2}}\cdots m_{k}!k!^{m_{k}}}
\\
&\qquad\qquad\qquad
\cdot\prod_{j=1}^{k}\left(\int_{0}^{T}\mu(T-s)F_{N,j}(s)ds\right)^{m_{j}},
\end{align*}
where
\begin{equation*}
F_{N,j}(t):=\frac{\partial^{j}}{\partial z^{j}}F_{N}(t)\bigg|_{z=0}.
\end{equation*}
>From \eqref{z2} it is clear that $F_{N,0}(t)=0$, and
\begin{equation*}
F_{N,1}(t)=\mathbb{E}\left[e^{-\int_{0}^{t}\ell_{1}h(s)ds}\right].
\end{equation*}
By applying Fa\`{a} di Bruno's formula again, we get for any $j\geq 2$
the following recursive equation:
\begin{align*}
F_{N,j}(t)&=j\frac{\partial^{j-1}}{\partial z^{j-1}}\mathbb{E}\left[e^{\int_{0}^{t}\ell_{1}h(s)(F_{N}(t-s)-1)ds}\right]\bigg|_{z=0}
\\
&=j\sum\frac{(j-1)!}{m_{1}!1!^{m_{1}}m_{2}!2!^{m_{2}}\cdots m_{j-1}!(j-1)!^{m_{j-1}}}
\\
&\qquad\qquad
\cdot
\mathbb{E}\left[e^{-\int_{0}^{T}\ell_{1}h(s)ds}\prod_{i=1}^{j-1}\left(\int_{0}^{t}\ell_{1}h(s)F_{N,i}(t-s)ds\right)^{m_{i}}\right].
\end{align*}
The proof is therefore completed.
\end{proof}

\section{Proof of Proposition~\ref{prop:LT-distri}}
\begin{proof}[Proof of Proposition~\ref{prop:LT-distri}]
The proof also relies on Fa\`{a} di Bruno's formula.
Note that for any $|z|<1$,
\begin{equation*}
\mathbb{E}\left[z^{\frac{1}{\delta}L_{T}}\right]
=\sum_{k=0}^{\infty}z^{k}\mathbb{P}(L_{T}=k\delta),
\end{equation*}
and thus,
\begin{equation*}
\mathbb{P}(L_{T}=k\delta)
=\frac{1}{k!}\frac{d^{k}}{dz^{k}}\mathbb{E}\left[z^{\frac{1}{\delta}L_{T}}\right]\bigg|_{z=0}.
\end{equation*}
For any $z\in\mathbb{C}$ with $|z|<1$,
\begin{equation*}
\mathbb{E}\left[z^{\frac{1}{\delta}L_{T}}\right]
=e^{\int_{0}^{T}\mu(T-s)(F_{L}(s)-1)ds},
\end{equation*}
where for any $0\leq t\leq T$,
\begin{equation*}
F_{L}(t)=\mathbb{E}\left[z^{\frac{1}{\delta}\ell_{1}}e^{\int_{0}^{t}\ell_{1}h(s)(F_{L}(t-s)-1)ds}\right]
=\sum_{k=1}^{\infty}z^{k}e^{\int_{0}^{t}k\delta h(s)(F_{L}(t-s)-1)ds}p_{k}.
\end{equation*}

By applying Fa\`{a} di Bruno's formula, we get:
\begin{align*}
\mathbb{E}\left[z^{\frac{1}{\delta}L_{T}}\right]\bigg|_{z=0}
&=e^{-\int_{0}^{T}\mu(T-s)ds}\sum\frac{k!}{m_{1}!1!^{m_{1}}m_{2}!2!^{m_{2}}\cdots m_{k}!k!^{m_{k}}}
\\
&\qquad\qquad\qquad
\cdot\prod_{j=1}^{k}\left(\int_{0}^{T}\mu(T-s)F_{L,j}(s)ds\right)^{m_{j}},
\end{align*}
where
\begin{equation*}
F_{L,j}(t):=\frac{\partial^{j}}{\partial z^{j}}F_{L}(t)\bigg|_{z=0}.
\end{equation*}
It is clear that $F_{L,0}(t)=0$, and
\begin{equation*}
F_{L,1}(t)=p_{1}e^{-\int_{0}^{t}\delta h(s)ds}.
\end{equation*}
By applying Fa\`{a} di Bruno's formula again, we get for any $j\geq 2$
the following recursive equation:
\begin{align*}
F_{L,j}(t)&=\frac{\partial^{j}}{\partial z^{j}}\sum_{k=0}^{\infty}z^{k}e^{\int_{0}^{t}k\delta h(s)(F_{L}(t-s)-1)ds}p_{k}\bigg|_{z=0}
\\
&=\sum_{k=0}^{\infty}\sum_{i=0}^{j}\binom{j}{i}\frac{d^{i}}{dz^{i}}z^{k}\bigg|_{z=0}
\frac{d^{j-i}}{dz^{j-i}}
e^{\int_{0}^{t}k\delta h(s)(F_{L}(t-s)-1)ds}p_{k}\bigg|_{z=0}
\\
&=\sum_{k=0}^{j}\binom{j}{k}k!
\frac{d^{j-k}}{dz^{j-k}}
e^{\int_{0}^{t}k\delta h(s)(F_{L}(t-s)-1)ds}p_{k}\bigg|_{z=0}
\\
&=\sum_{k=0}^{j}\binom{j}{k}k!
e^{-\int_{0}^{t}k\delta h(s)ds}p_{k}
\sum\frac{(j-k)!}{m_{1}!1!^{m_{1}}m_{2}!2|^{m_{2}}\cdots
m_{j-k}!(j-k)!^{m_{j-k}}}
\\
&\qquad\qquad\qquad\cdot
\prod_{i=1}^{j-k}\left(\int_{0}^{t}k\delta h(s)F_{L,i}(t-s)ds\right)^{m_{i}},
\end{align*}
where the summation is over $1\cdot m_{1}+2\cdot m_{2}+\cdots+(j-k)\cdot m_{j-k}=j-k$.
\end{proof}

\section{Derivations of Equations~\eqref{eq:cond-prob} and \eqref{eq:cond-prob2}}

Let us compute for a non-negative integer $k,$
\begin{equation*}
\mathbb{P}(N(t,t+T]=k|N_t=1,\ell_{1}=l_{1}),
\end{equation*}
where $N(t,t+T] = N_{t+T}- N_t.$
Our strategy is to first compute the probability generating function
of $N(t,t+T]$ conditional on $N_t=1$ and $\ell_{1}=l_{1}$.

Note that on $[0,t]$, the first jump $\tau_{1}$
has the probability density function
$\mu(\tau_{1})e^{-\int_{0}^{\tau_{1}}\mu(s)ds}$. Conditional
on the time of the first jump $\tau_{1}$, $N_t=1$
if and only if $N(\tau_{1},t]=0$,
which occurs with probability $e^{-\int_{\tau_{1}}^{t}(\mu(s)+h(s-\tau_{1})l_{1})ds}$
conditional on $\ell_{1}=l_{1}$.
Next, notice that conditional on there is
only one jump on $[0,t]$ and the time of the first jump being $\tau_{1}$
and conditional on the first jump size being $l_{1}$, the stochastic process $N(t,t+s]$
as a function of $s\in[0,T]$, is a Hawkes process with an exciting function $h(\cdot)$, i.i.d. jump sizes $\ell_{i}$
and the time-dependent baseline intensity $\mu(s)+h(s-\tau_{1})l_{1}$ at time $t<s<t+T$.

Hence, from the discussions above and the probability generating functions
we derived in \eqref{z1} and \eqref{z2} in the proof of Proposition~\ref{prop:mass},
we conclude that,
for any $z\in\mathbb{C}$ with $0\leq |z|\leq 1$, the probability generating function is given by
\begin{align}
H(z)&:=\mathbb{E}\left[z^{N(t,t+T]}|N_t=1,\ell_{1}=l_{1}\right]
\nonumber
\\
&=\frac{\int_{0}^{t}\mu(\tau_{1})e^{-\int_{0}^{\tau_{1}}\mu(s)ds}e^{-\int_{\tau_{1}}^{t}(\mu(s)+h(s-\tau_{1})l_{1})ds}
e^{\int_{t}^{t+T}(\mu(s)+h(s-\tau_{1})l_{1})(F_{N}(T+t-s)-1)ds}d\tau_{1}}
{\int_{0}^{t}\mu(\tau_{1})e^{-\int_{0}^{\tau_{1}}\mu(s)ds}e^{-\int_{\tau_{1}}^{t}(\mu(s)+h(s-\tau_{1})l_{1})ds}d\tau_{1}},
\label{pgfN}
\end{align}
where for any $0\leq t\leq T$,
\begin{equation*}
F_{N}(t)=z\mathbb{E}\left[e^{\int_{0}^{t}\ell_{1}h(s)(F_{N}(t-s)-1)ds}\right].
\end{equation*}
The probability generating function yields
\begin{equation*}
\mathbb{E}\left[z^{N(t,t+T]}|N_t=1,\ell_{1}=l_{1}\right]
=\sum_{k=0}^{\infty}z^{k}\mathbb{P}(N(t,t+T]=k|N_t=1,\ell_{1}=l_{1}),
\end{equation*}
and hence the Taylor expansion coefficient
of this generating function gives the probability mass function we need.
Hence, we can compute that
\begin{align*}
\mathbb{P}(N(t,t+T]=0|N(t)=1,\ell_{1}=l_{1})
&=H(0)=\frac{\int_{0}^{t}\mu(\tau_{1})e^{-\int_{0}^{\tau_{1}}\mu(s)ds}e^{-\int_{\tau_{1}}^{t+T}(\mu(s)+h(s-\tau_{1})l_{1})ds}
d\tau_{1}}
{\int_{0}^{t}\mu(\tau_{1})e^{-\int_{0}^{\tau_{1}}\mu(s)ds}e^{-\int_{\tau_{1}}^{t}(\mu(s)+h(s-\tau_{1})l_{1})ds}d\tau_{1}},
\end{align*}
and
\begin{align*}
&\mathbb{P}(N(t,t+T]=1|N(t)=1,\ell_{1}=l_{1})
\\
&=H'(0)=\left(\int_{0}^{t}\mu(\tau_{1})e^{-\int_{0}^{\tau_{1}}\mu(s)ds}
e^{-\int_{\tau_{1}}^{t}(\mu(s)+h(s-\tau_{1})l_{1})ds}d\tau_{1}\right)^{-1}
\\
&\qquad
\cdot\int_{0}^{t}\mu(\tau_{1})e^{-\int_{0}^{\tau_{1}}\mu(s)ds}e^{-\int_{\tau_{1}}^{t}(\mu(s)+h(s-\tau_{1})l_{1})ds}
e^{-\int_{t}^{t+T}(\mu(s)+h(s-\tau_{1})l_{1})ds}
\\
&\qquad\cdot
\int_{t}^{t+T}(\mu(s)+h(s-\tau_{1})l_{1})\frac{\partial}{\partial z}F_{N}(T+t-s)\bigg|_{z=0}dsd\tau_{1}
\\
&=\left(\int_{0}^{t}\mu(\tau_{1})e^{-\int_{0}^{\tau_{1}}\mu(s)ds}e^{-\int_{\tau_{1}}^{t}(\mu(s)+h(s-\tau_{1})l_{1})ds}d\tau_{1}\right)^{-1}
\\
&\qquad
\cdot\int_{0}^{t}\bigg[\mu(\tau_{1})e^{-\int_{0}^{\tau_{1}}\mu(s)ds}e^{-\int_{\tau_{1}}^{t+T}(\mu(s)+h(s-\tau_{1})l_{1})ds}
\\
&\qquad\qquad\qquad\qquad
\cdot
\int_{t}^{t+T}(\mu(s)+h(s-\tau_{1})l_{1})\mathbb{E} \left[e^{-\int_{0}^{T+t-s} \ell_1 h(u)du}\right]ds\bigg]d\tau_{1}.
\end{align*}


\section{Derivations of \eqref{321}}
We provide a direct proof for \eqref{321}.  It is obvious that $\mathbb{E}[\sigma^{(1)}_x] = \lceil x \rceil$, since $\sigma^{(1)}_x$ is the hitting time to level $x>0$ for a Poisson process with rate one. Hence, it suffices to show 
\begin{equation}\label{32}
\mathbb{E}[\sigma^{(3)}_x]  > \mathbb{E}[\sigma^{(2)}_x] = x +1 \quad \text{for $x > 0,$}
\end{equation}
where $\mathbb{E}[\sigma^{(i)}_x]$, $i=2,3$, are the expected hitting time to level $x>0$ for
the compound Poisson arrival with jump sizes $\ell_{i}^{(2)}$ being exponential and $\ell_{i}^{(3)}$ being hyper-exponential respectively. Let us write the density function of $\ell_{i}^{(3)}$  as
\begin{equation*}
c\lambda_{1}e^{-\lambda_{1}x}+(1-c)\lambda_{2}e^{-\lambda_{2}x},
\qquad
x>0,
\end{equation*}
where $0<c<1$ and $0 <\lambda_{1} <1< \lambda_{2}$
so that
\begin{equation*}
\mathbb{E}[\ell_{i}^{(3)}]
=\frac{c}{\lambda_{1}}+\frac{1-c}{\lambda_{2}}=1= \mathbb{E}[\ell_{i}^{(2)}].
\end{equation*}
Since the baseline intensity is one, we have $\{L_{t}^{(j)}- t: t \ge 0\}$ is a martingale for $j=2,3$, where $L_{t}^{(j)}$ is the point process with jump sizes $\ell_{i}^{(j)}$. Now we infer from
optional stopping theorem that
\begin{equation}\label{MEqn}
\mathbb{E}\left[\sigma_{x}^{(j)}\wedge M\right]=\mathbb{E}\left[L_{\sigma_{x}^{(j)}\wedge M}^{(j)}\right],
\qquad j=2,3,
\end{equation}
for any $M>0$. Note that $0\leq L_{\sigma_{x}^{(j)}\wedge M}^{(j)}\leq L_{\sigma_{x}^{(j)}}^{(j)}$.
By letting $M\rightarrow\infty$, we apply monotone convergence of the left hand side of
\eqref{MEqn} and dominated convergence theorem on the right hand side of \eqref{MEqn} and we get:
\begin{equation}\label{eq:E-sigma-x}
\mathbb{E}[\sigma_{x}^{(j)}]=x+\mathbb{E}\left[L_{\sigma_{x}^{(j)}}^{(j)}-x\right],
\qquad j=2,3,
\end{equation}
provided that $\mathbb{E}[L_{\sigma_{x}^{(j)}}^{(j)}]$ is finite for $j=2,3$.

Next, let us compute and estimate the expected overshoot $\mathbb{E}[L_{\sigma_{x}^{(j)}}^{(j)}-x]$. For $j=2,$ it is well known that for exponentially distributed $\ell_{i}^{(2)}$ with mean $1$, the overshoot is also exponentially distributed with mean $1$
and thus
\begin{equation}\label{eq:E-sigma2-x}
\mathbb{E}[\sigma_{x}^{(2)}]=x+\mathbb{E}\left[L_{\sigma_{x}^{(2)}}^{(2)}-x\right] = x+1.
\end{equation}
For $j=3,$ we note that for a hyper-exponentially distributed $\ell_{i}^{(3)}$ with mean $1$,
we can compute that for any $0<z<x$,
\begin{align} \label{eq:overshott}
\mathbb{E}\left[L_{\sigma_{x}^{(3)} }^{(3)}-x\big|L_{\sigma_{x}^{(3)}- }^{(3)}=x-z\right]
&=\mathbb{E}\left[\ell_{1}^{(3)}\big|\ell_{1}^{(3)}>z\right] \nonumber
\\
&=\frac{\int_{z}^{\infty}c\lambda_{1}(y-z)e^{-\lambda_{1}y}dy+\int_{z}^{\infty}(1-c)(y-z)\lambda_{2}e^{-\lambda_{2}y}dy
}{\int_{z}^{\infty}c\lambda_{1}e^{-\lambda_{1}y}dy+\int_{z}^{\infty}(1-c)\lambda_{2}e^{-\lambda_{2}y}dy}
\nonumber
\\
&=\frac{c\frac{1}{\lambda_{1}}e^{-\lambda_{1}z}+(1-c)\frac{1}{\lambda_{2}}e^{-\lambda_{2}z}}
{ce^{-\lambda_{1}z}+(1-c)e^{-\lambda_{2}z}}.
\end{align}
Notice that
\begin{equation*}
\frac{c\frac{1}{\lambda_{1}}e^{-\lambda_{1}z}+(1-c)\frac{1}{\lambda_{2}}e^{-\lambda_{2}z}}
{ce^{-\lambda_{1}z}+(1-c)e^{-\lambda_{2}z}}
\leq\frac{1}{\lambda_{1}}+\frac{1}{\lambda_{2}},
\end{equation*}
uniformly in $0<z<x$ and thus $\mathbb{E}[L_{\sigma_{x}^{(3)}}^{(3)}]\leq\frac{1}{\lambda_{1}}+\frac{1}{\lambda_{2}}+x<\infty$
is finite.
Moreover,
\begin{equation*}
\frac{c\frac{1}{\lambda_{1}}e^{-\lambda_{1}z}+(1-c)\frac{1}{\lambda_{2}}e^{-\lambda_{2}z}}
{ce^{-\lambda_{1}z}+(1-c)e^{-\lambda_{2}z}}
> 1,
\end{equation*}
if and only if
\begin{equation}\label{IneqHold}
c\left(\frac{1}{\lambda_{1}}-1\right)e^{-\lambda_{1}z}
>(1-c)\left(1-\frac{1}{\lambda_{2}}\right)e^{-\lambda_{2}z}.
\end{equation}
Since $\lambda_{2}>\lambda_{1}$ and $z>0$, the strict inequality \eqref{IneqHold} holds if we can show that
\begin{equation}\label{IneqHold2}
\frac{c}{\lambda_{1}}-c\geq 1-c-\frac{1-c}{\lambda_{2}},
\end{equation}
This holds and indeed we get the equality in \eqref{IneqHold2} due to $\mathbb{E}[\ell_{i}^{(3)}]=1$.
Hence, we can infer from \eqref{eq:overshott} that $\mathbb{E}[L_{\sigma_{x}^{(3)} }^{(3)}-x] >1$ when the jump size is hyper-exponentially distributed.
On combining with \eqref{eq:E-sigma-x} and \eqref{eq:E-sigma2-x}, we obtain \eqref{32}.

\end{document}